\documentclass[acmsmall]{acmart}
\usepackage{acm-ec-23}
\usepackage{booktabs} 
\usepackage[ruled, noend]{algorithm2e} 

\SetAlFnt{\small}
\SetAlCapFnt{\small}
\SetAlCapNameFnt{\small}
\SetAlCapHSkip{0pt}
\IncMargin{-\parindent}

\usepackage{stuff}
\usepackage{mathtools}
\usepackage{xcolor}
\usepackage{enumitem}
\usepackage{todonotes}
\usepackage{cleveref}

\title[Nearly Optimal Committee Selection For Bias Minimization]{Nearly Optimal Committee Selection For Bias Minimization}

\author{Yang Cai}
\authornote{Supported by a Sloan Foundation Research Fellowship and the NSF Award CCF-1942583 (CAREER).}
\authornote{Email: yang.cai@yale.edu}
\email{yang.cai@yale.edu}
\orcid{0000-0002-5426-1324}
\affiliation{%
    \institution{Yale University}
    \streetaddress{51 Prospect Street}
    \city{New Haven}
    \state{Connecticut}
    \country{USA}
    \postcode{06511-8937}
}

\author{Eric Xue}
\authornote{Email: ex3782@princeton.edu}
\email{ex3782@princeton.edu}
\orcid{0009-0001-3977-8173}
\affiliation{%
    \institution{Princeton University}
    \streetaddress{35 Olden Street}
    \city{Princeton}
    \state{New Jersey}
    \country{USA}
    \postcode{08540-5233}
}

\begin{abstract}
We study the model of metric voting initially proposed by \citet{Feldman_Mansour_Nisan_Oren_Tennenholtz_2020}.
In this model, experts and candidates are located in a metric space, and each candidate possesses a quality that is independent of her location.
An expert evaluates each candidate as  the candidate's quality less the distance between the candidate and the expert in the metric space. The expert votes for her favorite candidate. 
Naturally, the expert prefers candidates that are ``similar'' to herself, i.e., close to her location in the metric space, thus creating bias in the vote. 
The goal is to select a voting rule and a committee of experts to mitigate the bias. 
More specifically, given $m$ candidates, what is the minimum number of experts needed to ensure that the voting rule selects a candidate whose quality is at most $\varepsilon$ worse than the best one?
    
 Our first main result is a new way to select the committee using exponentially less experts compared to the method proposed in~\citet{Feldman_Mansour_Nisan_Oren_Tennenholtz_2020}.
Our second main result is a novel construction that  substantially improves the lower bound on the committee size. 
Indeed, our upper and lower bounds match in terms of $m$, the number of candidates, and $\varepsilon$, the desired accuracy, for general convex normed spaces, and differ by a multiplicative factor that only depends on the dimension of the underlying normed space but is independent of other parameters of the problem. 
We further extend the nearly matching upper and lower bounds to the setting in which each expert returns a ranking of her top $k$ candidates and we wish to choose $\ell$ candidates with cumulative quality at most $\varepsilon$ worse than that of the best set of $\ell$ candidates, settling an open problem of \citet{Feldman_Mansour_Nisan_Oren_Tennenholtz_2020}. 
Finally, we consider the setting where there are multiple rounds of voting. We show that by introducing another round of voting, the number of experts needed to guarantee the selection of an $\varepsilon$-optimal candidate becomes independent of the number of candidates. 
\end{abstract}

\begin{document}

\maketitle

\section{Introduction}

Preference for those similar to oneself is a natural source of bias in decision-making.
However, it may be problematic when the decision is of public consequence.
For example, consider the setting of academic hiring.
A dean is looking to hire a new faculty member for the computer science department.
Unfortunately, as an expert in another field, the dean cannot properly evaluate the qualifications of the candidates who applied for this position.
Thus, she decides to form a hiring committee composed of experts in the department to advise her decision-making.

However, these experts may be biased in their evaluations.
For example, an expert in one sub-field of computer science, say, operating systems, may implicitly prefer candidates who work in the same sub-field or closely related sub-fields, such as networks or compilers, over candidates who work in less related sub-fields, such as algorithms or human-computer interaction.
This preference may distort the expert's perception of a candidate's quality, despite the fact that the sub-field of computer science in which one works has little to do with whether one is qualified as a computer scientist on the whole.
Thus, the dean cannot simply choose the members of the hiring committee arbitrarily: a committee that has too many experts in one sub-field may be biased against candidates from another.
Rather, to mitigate these biases and hire the most qualified candidate, the dean could choose a committee that is diverse in the sub-fields of its members.

Despite this intuition, the question remains:  how exactly should the dean choose a committee of experts and aggregate their preferences in order to choose a candidate that is (nearly) socially optimal?
To study this question more concretely, we consider the model of metric voting initially proposed by \citet{Feldman_Mansour_Nisan_Oren_Tennenholtz_2020}.

In this model, experts and candidates are located in a metric space.
On an intuitive level, the location of an expert or a candidate in this metric space represents her features, e.g., how theoretical her research is, and the distance between an expert and a candidate measures how similar they are to each other.
In addition to a location, each candidate also possesses a quality that is independent of her location in the metric space and that only the experts can perceive.
In particular, the planner cannot observe candidates' qualities.
However, while an expert can accurately perceive the quality of a candidate, she (perhaps unintentionally) distorts and biases her evaluation of the candidate based on their degree of similarity.
More specifically, an expert evaluates a candidate as the distance between them detracted from the candidate's quality.
Thus, the greater the distance between an expert and a candidate, the more the expert underestimates the candidate's quality.
Based on her own evaluations of the candidates, each expert will vote for who she believes is the best candidate.
Importantly, this candidate may not be the true best candidate because each expert's evaluation of a candidate is distorted by her location in the metric space.
With knowledge of these votes, the locations of the experts, and the locations of the candidates, the planner chooses a candidate.
We study how to select a committee of experts and aggregate their preferences so that regardless of the candidates that arrive, the chosen candidate is no more than $\varepsilon > 0$ worse than the best candidate. 
This problem is termed the universal committee problem by \citet{Feldman_Mansour_Nisan_Oren_Tennenholtz_2020}. We also follow their convention and refer to the difference between the quality of the chosen candidate and the best candidate as the \emph{regret}.

For the metric space $([0,1]^d, \norm{\cdot}_p)$,\footnote{While \citet{Feldman_Mansour_Nisan_Oren_Tennenholtz_2020} do not explicit state the metrics for which their upper and lower bounds for the unit hypercube hold, the results stated here can be readily obtained from the ideas presented in their paper.} \citet{Feldman_Mansour_Nisan_Oren_Tennenholtz_2020} demonstrate a voting rule and a universal committee of size $O(m^d/\varepsilon^d)$\footnote{For simplicity, here, and for the remainder of this section, we suppress terms that depend solely on the dimension $d$.} such that for any set of $m$ candidates in $[0,1]^d$, the candidate selected by the voting rule is no more than $\varepsilon$ worse than the best candidate.
Moreover, for the same metric space, they show that any universal committee that guarantees at most $\varepsilon$ regret for any set of $m$ candidates is of size at least $\Omega(\max\{m, 1/\varepsilon^d\})$.
However, these bounds paint a picture of the optimal committee that is far from complete.
In particular, while the upper bound depends simultaneously on the number of candidates $m$ and the desired regret bound $\varepsilon$, the lower bound only depends on these parameters separately.
Moreover, there is an exponential gap in the upper and lower bounds: while the upper bound grows linearly with $m^d$, the lower bound either grows linearly with $m$ or does not grow with $m$ at all.

\subsection{Our Results}
In this paper, we obtain results that tell a nearly complete story of the optimal committee for general convex normed spaces.
More specifically, we show that for general convex normed spaces, there exists a universal committee that guarantees a regret of at most $\varepsilon > 0$ for any $m$ candidates whose size depends linearly on the number of candidates $m$ and inversely on $\varepsilon^d$ and that these dependencies are tight in a strong sense: any universal committee with this regret guarantee is smaller by at most a multiplicative factor that only depends on the dimension of the normed space and is independent of all other parameters of the problem.
To contrast our results with the results of \citet{Feldman_Mansour_Nisan_Oren_Tennenholtz_2020}, we obtain as a corollary to our main results that for $([0,1]^d, \norm{\cdot}_p)$, the smallest universal committee that guarantees at most $\varepsilon$ regret for any $m$ candidates is of size $\Theta(m/\varepsilon^d)$.\footnote{Recall that we are suppressing terms that depend solely on the dimension $d$ of the underlying space.}

The key idea behind the exponential improvement in our upper bound is that a universal committee need not cover the metric space as finely as \citet{Feldman_Mansour_Nisan_Oren_Tennenholtz_2020} suggest. More specifically, \citet{Feldman_Mansour_Nisan_Oren_Tennenholtz_2020} chooses the committee to be an $O(\varepsilon/m)$-cover. We show that 
it suffices to first place an $O(\varepsilon)$-cover on the metric space, then construct a spanning tree over the elements in the cover and choose experts to form an $O(\varepsilon/m)$-cover of the spanning tree.
That is, for $d$-dimensional spaces, our committee covers a \textit{one-dimensional} subspace, while the committee of \citet{Feldman_Mansour_Nisan_Oren_Tennenholtz_2020} covers the entire $d$-dimensional space.
The $O(\varepsilon)$-cover of the metric space prevents any candidate who did not receive a vote from being much better than a candidate who did, while the $O(\varepsilon/m)$-cover of the spanning tree allows us to compare the qualities of faraway candidates without incurring too much error.
Moreover, by placing experts this way, while the size of the committee depends inversely on $\varepsilon^d$, the dependence on the number of candidates $m$ is only linear.
Thus, our committee is significantly sparser than that of \citet{Feldman_Mansour_Nisan_Oren_Tennenholtz_2020}.
See~\Cref{theorem:general-upper-bound} for the formal statement.

To demonstrate a nearly matching lower bound, we show that if a universal committee guarantees $\varepsilon$ regret for any $m$ candidates, then each ball of an $\Omega(\varepsilon)$-packing must contain $\Omega(m)$ experts.
Otherwise, we can place candidates within the ball of radius $\Omega(\varepsilon)$ with too few experts in such a way so that the candidates with high quality under one quality vector are of low quality under another, yet both quality vectors induce identical votes from the experts.
We achieve this phenomenon by placing a number of candidates that is independent of $m$ and $\varepsilon$ on the boundary of the ball per expert in the interior of the ball.
These candidates hide the goings-on within the ball from the view of experts outside of the ball. See \Cref{theorem:general-lower-bound} for the formal statement.

While we present our results for a general dimension $d$, we see the problem as already interesting for small $d$, such as $d=2$ or $d=3$, where $1/\varepsilon^d$ is not too large and assembling a committee of such size may be feasible.
Even for such $d$, the universal committee that leads to our upper bound is already substantially sparser than the universal committee put forth by \citet{Feldman_Mansour_Nisan_Oren_Tennenholtz_2020}: for $d=2$, we obtain a significant quadratic improvement and for $d=3$, we obtain a significant cubic improvement.
Because of their sparseness, we believe that our committees remain potentially feasible for small $d$.
This realization would not be possible without a complete understanding of the optimal dependence on the relevant parameters of the problem.

We highlight that since the covering and packing numbers are a multiplicative constant to the dimension apart, the gap between our upper and lower bounds is no larger. We extend our nearly matching upper and lower bounds to the setting in which each expert returns a ranking of her top $k$ candidates and we wish to choose $\ell$ candidates with cumulative quality at most $\varepsilon$ worse than that of the best set of $\ell$ candidates, settling an open problem of \citet{Feldman_Mansour_Nisan_Oren_Tennenholtz_2020}.
We show the size of the optimal committee depends inversely on $k$ and linearly on $\ell^d$.
Our main results---Theorems \ref{theorem:general-upper-bound} and \ref{theorem:general-lower-bound}---are stated for general $k$ and $\ell = 1$. 
It is straightforward to generalize the $\ell = 1$ case to general $\ell$, and we give a formal statement and a proof sketch for general $k$ and general $\ell$ in Theorem \ref{theorem:multi-winner} and its subsequent discussion.

In a separate direction, we consider the setting in which there are multiple rounds of voting.
Interestingly, with only two rounds of voting, it is possible to eliminate the dependence of the committee size on the number of candidates $m$.
The key idea here is to first screen the candidates so that the number of candidates to choose from is bounded by a number that is solely dependent on the desired regret bound $\varepsilon$ and the dimension of the underlying space $d$.
Then, we can apply our upper bound for general convex normed spaces to bound the size of the committee that will select from the remaining candidates.
By choosing the screening committee to be sufficiently large, we can ensure that no candidate who does not pass the first round of voting is significantly better than a candidate who does pass. See \Cref{theorem:multi-round-voting} for the formal statement.





\section{Related Work}

Similar to \citet{distortion-cardinal-preferences}, \citet{optimal-social-choice-utilitarian}, \citet{metric-voting-1}, and \citet{subset-selection-implicit-utilitarian}, our work adopts a utilitarian view of voting in which utility functions that are unobservable by the planner underlie the preferences of voters.
However, unlike them, the planner's goal is not to maximize the social welfare but to choose a candidate with an approximately optimal quality, and the purpose of the voters is to provide information on the qualities of candidates rather than information on who is most favored.


Among works that adopt a utilitarian view of voting, those that study metric voting \citep{metric-voting-1, metric-voting-2, Skowron_Elkind_2017, munagala-wang, Kempe_2020, halpern} are particularly relevant to our work.
As in our work, voters and candidates are located in a metric space whose existence is unknown to them.
The preferences of voters are induced by the underlying metric: each voter prefers the candidates closer to them in the metric space over those that are further away.
The goal in these works is to choose a candidate that approximately maximizes social welfare.
Unlike this line of work, while voters in our work tend to prefer closer candidates over further ones, the quality of a candidate also factors into each voter's preference.
Moreover, on an interpretive level, the distance between a voter and a candidate represents the voter's implicit bias against the candidate rather than the candidate's social cost.
Finally, our planner is aware of the underlying metric space and wants to locate voters in the metric space in a way that guarantees that a nearly optimal candidate is chosen.


\section{Preliminaries}\label{section:preliminaries}


\subsection{Model}


We study a generalized version of the model of voting under metric preferences proposed in \citet{Feldman_Mansour_Nisan_Oren_Tennenholtz_2020}.
In this model, candidates of varying quality and experts are located in a metric space. 
Experts evaluate each candidate as a function of the candidate's quality and the distance between them according to the associated metric.
Let $(\Theta ,d(\cdot, \cdot))$ denote the metric space, and let $[m] \coloneqq \{1, \dots, m\}$ denote the set of candidates.
Associate with each candidate $j \in [m]$ a location $c_j \in \Theta$ and a quality $q_j \in \RR$.
We often let $C \coloneqq \{c_1, \dots, c_m\}$ denote the set of candidate locations and $q = (q_1, \dots, q_m) \in \RR^m$ denote the vector of candidate qualities.

To evaluate these candidates, we select a set of experts, denoted $[n] \coloneqq\{1, \dots, n\}$.
For each expert $i \in [n]$, let $e_i$ denote her location in the metric space.
We often let $E \coloneqq \{e_1, \dots, e_n\}$ denote the locations of the experts.
In the model proposed by \citet{Feldman_Mansour_Nisan_Oren_Tennenholtz_2020}, expert $i$ votes for candidate $j$ if and only if $j$ maximizes the perceived quality $q_j - d(e_i, c_j)$.
To generalize this model, we consider what happens when each expert returns a ranking of her top $k$ candidates.
For each expert $i$, let $\succ_i^k$ denote her top $k$ ranking.
For $h \in [k]$, let ${\succ_i^k}(h) \in \arg\max_{j \in [m] \setminus \{{\succ_i^k}(\ell)\}_{\ell = 1}^{h-1}} q_j - d(e_i, c_j)$ denote expert $i$'s $h$-th most preferred candidate. Note that we allow arbitrary tie-breaking rules.
We say that expert $i$ prefers candidate $j$ over candidate $j'$ and write $j \succ_i^k j'$ if $q_j - d(e_i, c_j) \geq q_{j'} - d(e_i, c_{j'})$ and $j = {\succ_i^k}(h)$ for some $h \in [k]$.

We often use $\succ^k \coloneqq \{\succ_i^k\}_{i \in [n]}$ to refer to the set of reported rankings and $V(\succ^k) \coloneqq \{{\succ_i^k}(1)\}_{i \in [n]}$ to refer to the set of candidates who rank first for some expert.
We say that $\hat{q} \in \RR^m$ is \emph{consistent} with $\succ^k$ given $C \coloneqq \{c_j\}_{j \in [m]}$ and $E \coloneqq \{e_i\}_{i \in [n]}$ if for all $i, j, j'$, $j \succ_i^k j' \implies \hat{q}_j - d(e_i, c_j) \geq \hat{q}_{j'} - d(e_i, c_{j'})$.
Conversely, we say that $\succ^k$ is \emph{inducible} given $C$ and $E$ if there exists $\hat{q} \in \RR^m$ consistent with $\succ^k$.

We aggregate the experts' top $k$ rankings and select a candidate according to a voting rule $f$ that maps candidate locations, expert locations, and top $k$ rankings to a candidate in $[m]$.\footnote{We also consider voting rules that select $\ell \leq k$ candidates in Section . Note that $\ell \leq k$ is the right interval to consider: if we hired $\ell > k$ candidates, then we would have no guarantees on their qualities. Consider two worlds with $\ell+1$ candidates. In both worlds, the first $\ell - 1 \geq k$ candidates possess qualities that are much higher than the $\ell$-th and the $(\ell+1)$-st candidates, so each expert will return the same $k$ candidates, the two worlds will be indistinguishable, and the same $\ell$ candidates will be chosen. However, in one world, have the $\ell$-th candidate be much better than the $(\ell+1)$-st candidate, and in the other world, have the opposite be true. Since we are hiring $\ell$ candidates, one of these two candidates will be chosen, so in one of the worlds, the worst candidate by far will be chosen.}
We would like to design a voting rule that always selects the most qualified candidate, but unfortunately, no such rule exists.
Thus, \citet{Feldman_Mansour_Nisan_Oren_Tennenholtz_2020} design a voting rule that minimizes regret over all quality vectors that are consistent with the experts' votes.
Here, the regret of choosing candidate $j^*$ under quality vector $q$ is defined as $\max_j q_j - q_{j^*}$.

\subsection{$k$-Approval Ballots}

Another natural generalization of the model proposed by \citet{Feldman_Mansour_Nisan_Oren_Tennenholtz_2020} is to have each expert return an (unranked) set of her top $k$ candidates (i.e., $k$-approval ballots) instead of a ranking of them.
More formally, each expert $i$ returns $\{{\succ_i}(h)\}_{h \in [k]}$ instead of $\succ_i$.
We focus on the setting in which each expert returns a ranking because if each expert only returns a set of $k$ candidates, then there exists a set of candidates such that every committee incurs constant regret.
That is, increasing the size of the committee will not decrease the incurred regret.

Suppose each expert returns $k=2$ approval ballots, and consider the following set of candidates on the unit interval equipped with the Euclidean metric. 
Locate candidates at 0, 1/2, and 1. In the first world, the qualities of the candidates at 0 and 1 are 1/2, while the quality of the candidate at 1/2 is 0. 
In the second, the qualities are flipped: the quality of the candidates at the extremes are 0, while the quality of the candidate in the middle is 1/2. 
Suppose each expert breaks ties by choosing the candidate closest to her (and if there are still ties, then she chooses a candidate arbitrarily). 
In both worlds, any expert in $[0,1/2]$ will return the candidates at 0 and 1/2, while any expert in $[1/2,1]$ will return the candidates at 1/2 and 1, so a voting rule cannot distinguish between the two worlds and chooses the same candidate in each, incurring a regret of 1/2 in one of the worlds regardless of how the committee is designed.\footnote{If each expert informs the planner of her favorite and second favorite candidates, then the planner can distinguish between the two worlds, e.g., in the first world, each expert in $[0,1/2]$ prefers the candidate at 0 over the candidate at 1/2; in the second world, their preferences are flipped.} 
It is not hard to extend this example to $k \geq 2$ approval ballots: duplicate the candidate at 1/2 $k-2$ times; in both worlds, each expert will vote for the $k-1$ candidates at 1/2 and the candidate at the extreme closest to her. 
Moreover, if our goal were to hire $2 \leq \ell \leq k$ candidates instead of just one, then a similar example shows that we incur a cumulative regret of 1/2: locate $\ell - 2$ candidates of extremely high quality anywhere in the unit interval and duplicate the candidate at 1/2 $k-\ell$ times; in both worlds, each expert will vote for the $\ell - 2$ candidates of high quality, the $k-\ell+1$ candidates located at 1/2, and the candidate at the extreme closest to her.
A voting rule must pick the $\ell - 2$ candidates of high quality to have low regret.
The remaining two candidates will be among the candidates at 0, 1/2, and 1, and at least one of two will cause the voting rule to incur a regret of 1/2 in at least one of the worlds.

\subsection{Generalized Minimal Regret Voting Rule}

We generalize the minimal regret voting rule proposed by \citet{Feldman_Mansour_Nisan_Oren_Tennenholtz_2020}.
Note that for any given candidate locations $C \coloneqq \{c_j\}_{j \in [m]}$, expert locations $E \coloneqq \{e_i\}_{i \in [n]}$, and rankings $\succ^k = \{\succ^k_i\}_{i \in [n]}$, there may be infinite quality vectors consistent with $\succ^k$.
More specifically, any $\hat{q} \in \RR^m$ such that $\hat{q}_j - d(e_i, c_j) \geq \hat{q}_{j'} - d(e_i, c_{j'})$ for all $i \in [n]$ and $j, j' \in [m]$ such that $j \succ_i^k j'$ is consistent with $\succ^k$.
Since we cannot distinguish between these quality vectors when we only have access to $\succ^k$, if we wish to minimize our regret, we must select the candidate that has the smallest regret when the underlying qualities are the worst possible for selecting her.
Thus, the minimal regret voting rule chooses $g \in [m]$ that minimizes the maximum of
\begin{equation*}\label{primal_lp}
\begin{array}{ll@{}ll}
\max_{\hat{q}}  & & \hat{q}_h - \hat{q}_g &  \\
\text{s.t.} & & \hat{q}_j - d(e_i, c_j) \geq \hat{q}_{j'} - d(e_i, c_{j'}) & \forall\: i \in [n], j, j' \in [m] \text{ s.t. } j \succ_i^k j'
\end{array}\tag{LP 1}
\end{equation*}
over all $h \not= g$.
We will refer to such a candidate as the minimal regret candidate.
Taking the dual of \ref{primal_lp} reveals that the optimal primal objective is the length of the shortest $g \to h$ path in the following weighted graph:

\begin{definition}\label{def:minimal regret voting graph}
Given the locations of $m$ candidates $C \coloneqq \{c_j\}_{j \in [m]}$, the locations of $n$ experts $E \coloneqq \{e_i\}_{i \in [n]}$, and their reported rankings $\succ^k \coloneqq \{\succ_i^k\}_{i \in [n]}$, define $G(C, E, \succ^k)$ to be the weighted directed graph with candidates as vertices and an edge from one candidate $j$ to another candidate $j'$ with weight $\min_{i : j \succ_i^k j'} d(e_i, c_{j'}) - d(e_i, c_{j})$ if and only if there exists an expert who prefers $j$ over $j'$.
Formally,
\[
    G(C, E, \succ^k) \coloneqq \left([m], \{(j, j') \in [m]^2 : \exists\: i \text{ s.t. } j \succ_i^k j'\}, w \right) \text{ where } w(j, j') \coloneqq \min_{i : j \succ_i^k j'} d(e_i, c_{j'}) - d(e_i, c_{j})
\]
\end{definition}

\begin{lemma}\label{lemma:reduction}
Let $C:= \{c_1,\ldots, c_m\}$ denote the set of locations for the $m$ candidates, $E:= \{e_1,\ldots,e_n\}$ denote the set of locations for the $n$ experts, and $\succ^k \coloneqq \{\succ_1^k, \dots, \succ_n^k\}$ denote an inducible set of rankings.
For any candidates $h\neq g$, the optimal objective value of the corresponding \ref{primal_lp} is $d_{G}(g,h)$ where $d_G(j, j')$ denotes the length of the shortest $j \to j'$ path in $G(C, E, \succ^k)$. Note that if there is no path from $g$ to $h$, then $d_{G}(g,h)=+\infty$. 
\end{lemma}

By weak duality, we obtain the following corollary of Lemma \ref{lemma:reduction}.

\begin{corollary}\label{corollary:reduction}
Let $C= \{c_1,\ldots, c_m\}$ denote the set of locations for the $m$ candidates, $q = (q_j)_{j \in [m]}$ denote the true but unknown qualities of the candidates, $E= \{e_1,\ldots,e_n\}$ denote the set of locations for the $n$ experts, and $\succ^k = \{\succ_1^k, \dots, \succ_n^k\}$ denote the reported rankings of the experts.
The regret of choosing candidate $j^* \in [m]$ is
$\max_{j \in [m]} q_j - q_{j^*} \leq \max_{j \in [m]} d_{G}(j^*, j)$
\end{corollary}

By Lemma \ref{lemma:reduction}, the minimal regret candidate is the graph center of $G(C, E, \succ^k)$.
Thus, to compute the minimal regret candidate, we can simply compute all-pairs shortest paths in $G(C, E, \succ^k)$ (rather than solve \ref{primal_lp} for all $g, h \in [m]$) and compute $\arg\min_{j \in [m]} \max_{j' \not= j} d_G(j, j')$.
We give the generalized minimal regret voting rule in Algorithm \ref{algorithm:minimal-regret}.
Since it takes $O(nmk)$ time to construct $G(C, E, \succ^k)$ and $O(m^3)$ time to compute all-pairs shortest paths, the running time of Algorithm \ref{algorithm:minimal-regret} is $O(nmk + m^3)$.

\begin{algorithm}[t]
    \SetAlgoNoLine
    \DontPrintSemicolon
    \KwIn{$C \coloneqq \{c_1, \dots, c_m\}$ is a set of $m$ candidate locations,
        $E \coloneqq \{e_1, \dots, e_n\}$ is a set of $n$ expert locations,
        $\succ^k \coloneqq \{\succ^k_1, \dots, \succ^k_n\}$ is a inducible set of rankings}
    \KwOut{a minimal regret candidate} 
    
    $G \leftarrow ([m], \{(j, j') \in [m]^2: j \not= j'\}, w(\cdot))$ where $w(j, j') \coloneqq +\infty$ for all $j \not= j'$\;
    \For{$i = 1, \dots, n$}{
        \For{$r = 1, \dots, k$}{
            \For{$r' = r+1, \dots, k$}{
                \If{$d(e_i, c_{{\succ_i^k}(r')}) - d(e_i, c_{{\succ_i^k}(r)}) < w({\succ_i^k}(r), {\succ_i^k}(r'))$}{
                    $w({\succ_i^k}(r), {\succ_i^k}(r')) \leftarrow d(e_i, c_{{\succ_i^k}(r')}) - d(e_i, c_{{\succ_i^k}(r)})$\;
                }
            }
            \For{$j' \in [m] \setminus \{{\succ_i^k}(h)\}_{h=1}^k$}{
                \If{$d(e_i, c_{j'}) - d(e_i, c_{{\succ_i^k}(r)}) < w({\succ_i^k}(r), j')$}{
                    $w({\succ_i^k}(r), j') \leftarrow d(e_i, c_{j'}) - d(e_i, c_{{\succ_i^k}(r)})$\;
                }
            }
        }
    }
    
    Solve all-pairs shortest paths in $G$\;
    \KwRet{$\arg\min_{j} \max_{j' \not= j} d_G(j, j')$}
    
    \caption{Generalized Minimal Regret Voting Rule}
    \label{algorithm:minimal-regret}
\end{algorithm}

The goal of the planner is now to construct a committee that guarantees low regret, say, a regret of at most $\varepsilon$, under the generalized minimal regret voting rule regardless of the $m$ candidates that may appear.
By Corollary \ref{corollary:reduction}, this amounts to constructing a committee $E$ so that for all candidate locations $C$ and inducible rankings $\succ^k$, there always exists a candidate whose distance to the other candidates in $G(C, E, \succ^k)$ is at most $\varepsilon$.
If such a candidate exists, then the regret of selecting the minimal regret candidate, i.e., the graph center of $G(C, E, \succ^k)$, is at most $\varepsilon$.

\subsection{Covers and Packings}

Our upper and lower bounds for general convex normed spaces are achieved using coverings and packing and are expressed via covering and packing numbers.

\begin{definition}
Let $(V, \norm{\cdot})$ be a normed vector space and $\Theta \subseteq V$.
$X \subseteq \Theta$ is an (internal) $\varepsilon$-cover if for all $\theta \in \Theta$, there exists $x \in X$ such that $\norm{\theta - x} \leq \varepsilon$.
Let $N(\Theta, \norm{\cdot}, \varepsilon)$ denote the minimum number of $\varepsilon$-balls needed to cover $\Theta$.
\end{definition}

\begin{definition}
Let $(V, \norm{\cdot})$ be a normed vector space and $\Theta \subseteq V$.
$X \subseteq \Theta$ is an $\varepsilon$-packing if for all $x \not= x' \in X$, $\norm{x - x'} \geq \varepsilon$.
Let $M(\Theta, \norm{\cdot}, \varepsilon)$ denote the size of a $\varepsilon$-packing of $\Theta$ of greatest cardinality
\end{definition}

\section{General Upper Bound}\label{section:general-upper-bound}

\begin{theorem}\label{theorem:general-upper-bound}
Let $(V, \norm{\cdot})$ be a normed vector space and $\Theta \subseteq V$ be convex.
If each expert returns a ranking of her top $k$ candidates, then there exists a universal committee of size $(8(m-1)/k + 1) (N\left(\Theta, \norm{\cdot}, \varepsilon/2\right)-1)$ such that for any set of $m$ candidates in $\Theta$, the regret of selecting the minimal regret candidate is at most $\varepsilon > 0$.
\end{theorem}

At a high level, the universal committee underlying Theorem \ref{theorem:general-upper-bound} is constructed as follows.
First, place an $(\varepsilon/2)$-cover on $\Theta$.
There will exist a spanning tree of the elements in the cover in which each edge is of length at most $\varepsilon$.
On each edge of this spanning tree, uniformly place $O(m/k)$ experts.
Because $\Theta$ is convex, these experts will lie inside $\Theta$.
Moreover, by placing experts in this way, we have constructed an $O(k\varepsilon/m)$-cover of the spanning tree.
Since there are $N(\Theta, \norm{\cdot}, \varepsilon/2) - 1$ edges in the spanning tree, this committee will be of size $O((m/k) N(\Theta, \norm{\cdot}, \varepsilon/2))$.

Let $V$ be the set of candidates who are some expert's favorite. The $(\varepsilon/2)$-cover on $\Theta$ will prevent any candidate that is not in $V$ 
from being more that $\varepsilon/2$ better than a candidate who is, while the $O(k\varepsilon/m)$-cover of the spanning tree will allow us to distinguish the quality of two candidates in $V$ up to an error of $O(\varepsilon)$. More specifically, we will see that because each expert returns a ranking of her top $k$ candidates, the graph $G(C,E,\succ^k)$ constructed in the generalized minimal regret voting rule will admit a path with $O(m/k)$ edges between any two candidates in $V$.  These paths will allow us to distinguish the quality of any two candidates who ranked first for some expert up to an error of $O(\varepsilon)$.
Thus, we can choose an candidate whose regret is at most $O(\varepsilon)$.

In contrast, \citet{Feldman_Mansour_Nisan_Oren_Tennenholtz_2020} construct a universal committee by placing an $(\varepsilon/(4m))$-cover on $\Theta$.
While their upper bound of the committee size has an exponentially worse dependence on the number of candidates,\footnote{The number of experts they need scales linearly in $m^d$ while our universal committee size only scale linearly in $m$ regardless of the dimension.} their construction works for any $\Theta$ that is connected (in particular, $\Theta$ need not be convex). 

While some of the techniques behind our upper bound result, such as using the graph associated with the dual LP to upper bound a particular committee's regret, are present in \citet{Feldman_Mansour_Nisan_Oren_Tennenholtz_2020}, our realization that an $O(\varepsilon/m)$-cover of a spanning tree whose vertices form an $O(\varepsilon)$-cover is precisely the construction that yields a tight upper bound and our observation that more informative preferences correspond to some notion of graph diameter are non-trivial contributions.
Indeed, \citet{Feldman_Mansour_Nisan_Oren_Tennenholtz_2020} list understanding what happens when each expert returns more than just her favorite candidate as an open problem.

\subsection{Proof of Theorem \ref{theorem:general-upper-bound}}




In our proof of Theorem \ref{theorem:general-upper-bound}, we will often refer to the following unweighted subgraph of $G(C, E, \succ^k)$ and the following (generalized) notion of graph diameter.

\begin{definition}\label{def:delta minimal regret voting graph}
Given the locations of $m$ candidates $C = \{c_j\}_{j \in [m]}$, the locations of $n$ experts $E = \{e_i\}_{i \in [n]}$, and their reported rankings $\succ^k = \{\succ_i^k\}_{i \in [n]}$, define $G(C, E, \succ^k)[\delta]$ to be the unweighted directed subgraph of $G(C, E, \succ^k)$ with candidates as vertices and an edge from one candidate $j$ to another candidate $j'$ if and only if there exists two experts $i$ and $i'$ who are within $\delta$ of each other in the underlying metric space $(\Theta, \metric{\cdot}{\cdot})$ yet differ in their preference between $j$ and $j'$.
Formally,
\[
    G(C, E, \succ^k)[\delta] \coloneqq ([m], \{(j, j') \in [m]^2 : \exists\:i, i' \in [n] \text{ s.t. } j \succ_i^k j', j \prec_{i'}^k j', \metric{e_i}{e_{i'}} \leq \delta\})\footnotemark{}
\]
\end{definition}
\footnotetext{Note that whenever $(j, j')$ is an edge in $G(C, E, \succ^k)[\delta]$, $(j', j)$ is also an edge in the graph. Nonetheless, we chose to define $G(C, E, \succ^k)[\delta]$ as a directed graph to emphasize its relation to $G(C, E, \succ^k)$.}

\begin{definition}
Let $G = (V, E)$ be a directed graph, $d_G(u, v)$ be the length of the shortest $u \to v$ path in $G$, and $S \subseteq V$.
Define $\diam{G, S} \coloneqq \max_{u, v \in S} d_G(u, v)$ to be the length of the longest shortest path between any two vertices in $S$.
Note that while we only consider paths with endpoints in $S$, we allow these paths to visit vertices not in $S$.
\end{definition}

\begin{lemma}\label{lemma:main-lemma}
Let $C= \{c_1,\ldots, c_m\}$ denote the set of locations for the $m$ candidates, 
$E= \{e_1,\ldots,e_n\}$ denote the set of locations for the $n$ experts, and $\succ^k = \{\succ_1^k, \dots, \succ_n^k\}$ denote the reported rankings of the experts.
Let $\hat{q} \in \RR^m$ be \textbf{any} quality vector consistent with $\succ^k$.
If $C \subseteq \cup_{i \in E} B(e_i, r)$ for some $r > 0$ 
 and $\diam{G(C, E, \succ^k)[\delta], V(\succ^k)} < +\infty$ for some $\delta > 0$, then, for all $j^* \in \arg\max_{j \in V(\succ^k)} \hat{q}_j$,
\[
    \max_{j \in [m]} d_{G}(j^*, j) \leq r + 2 \delta \cdot \diam{G(C, E, \succ^k)[\delta], V(\succ^k)}
\]
where $G(C, E, \succ^k)$ is the graph constructed in the generalized minimal regret voting rule (Definition~\ref{def:minimal regret voting graph}), and $d_{G}(j, j')$ is the distance of the shortest $j \to j'$ path in $G(C, E, \succ^k)$.
\end{lemma}

An implication of Corollary \ref{corollary:reduction} and Lemma \ref{lemma:main-lemma} is that given candidate locations $C$, expert locations $E$, and their rankings $\succ^k$, the regret of choosing any candidate $j^* \in \arg\max_{j \in V(\succ^k)} \hat{q}_j$ for some proxy quality vector $\hat{q} \in \RR^m$ consistent with $\succ^k$ is at most $r + 2 \delta \diam{G(C, E, \succ^k)[\delta], V(\succ^k)}$.
If the committee $E$ is constructed so that $r = O(\varepsilon)$ and $\diam{G(C, E, \succ^k)[\delta], V(\succ^k)} = O(\frac{\varepsilon}{\delta})$ for all $C$ and inducible $\succ^k$, then the regret of choosing a candidate such as $j^*$ is at most $O(\varepsilon)$.
Thus, to compute a low regret candidate, one need not appeal to $G(C, E, \succ^k)$ at all.
Rather, one can simply solve a primal feasibility problem and choose the candidate with the greatest proxy quality among those who rank first for some expert.
Somewhat surprisingly, consistency with experts' rankings is sufficient to guarantee that the true quality of the candidate with the greatest proxy quality among those who rank first for some expert is within $O(\varepsilon)$ of the greatest true quality among all candidates.

\begin{proof}[Proof of Lemma \ref{lemma:main-lemma}]
Recall that $G(C, E, \succ^k)$ as defined in the generalized minimal regret voting rule (Definition~\ref{def:minimal regret voting graph}).
\[
    G(C, E, \succ^k) \coloneqq \left([m], \left\{(j, j') \in [m]^2 : \exists\: i \in [n] \text{ s.t. } j \succ_{i}^k j' \right\}, w\right)
\]
where $w(j, j') = \min_{i \in E : j \succ_{i}^k j'} \metric{e_{i}}{c_{j'}} - \metric{e_{i}}{c_{j}}$.

Consider any pair of candidates $j, j' \in V(\succ^k)$.
Since $\diam{G(C, E, \succ^k)[\delta], V(\succ^k)} < +\infty$, there exists a $j \to j'$ path in $G(C, E, \succ^k)[\delta]$ consisting of at most $\diam{G(C, E, \succ^k)[\delta], V(\succ^k)}$ edges.
Let $j = j_1 \to \dots \to j_S = j'$ where $S \leq \diam{G(C, E, \succ^k)[\delta], V(\succ^k)} + 1$ denote this path.
By definition of $G(C, E, \succ^k)[\delta]$, for all $s \in [S-1]$, there exists $i, i' \in [n]$ such that $j_{s} \succ_{i}^k j_{s+1}$, $j_{s} \prec_{i'}^k j_{s+1}$, and $\metric{e_{i}}{e_{i'}} \leq \delta$.
Thus,
\[
    \hat{q}_{j_{s+1}} - \hat{q}_{j_{s}} \geq \metric{e_{i'}}{c_{j_{s+1}}} - \metric{e_{i'}}{c_{j_{s}}} \geq \metric{e_{i}}{c_{j_{s+1}}} - \metric{e_{i}}{c_{j_s}} - 2\metric{e_{i}}{e_{i'}} \geq w(j_s, j_{s+1}) - 2\delta
\]

Here, the first inequality follows from the fact that $j_{s} \prec_{i'}^k j_{s+1}$.
The second inequality follows from the triangle inequality.
The third inequality follows from the definition of $w$ and the fact that $j_{s} \succ_{i}^k j_{s+1}$. The chain of inequalities shows that the difference of the proxy qualities of candidate $j_{s+1}$ and $j_s$ can serve as an upper bound of $w(j_s,j_{s+1})$.
Since $S \leq \diam{G(C, E, \succ^k)[\delta], V(\succ^k)} + 1$, it follows that
\[
    \hat{q}_{j'} - \hat{q}_{j} = \sum_{s=1}^{S-1} \hat{q}_{j_{s + 1}} - \hat{q}_{j_s} \geq \sum_{s=1}^{S-1} w(j_s, j_{s+1}) - 2\delta \cdot \diam{G(C, E, \succ^k)[\delta], V(\succ^k)}
\]
Moreover, since $G(C, E, \succ^k)[\delta]$ is a subgraph of $G(C, E, \succ^k)$, $j_1 \to \dots \to j_S$ is a $j \to j'$ path in $G(C, E, \succ^k)$ as well, so $\sum_{s=1}^{S-1} w(j_s, j_{s+1}) \geq d_{G}(j, j')
$.
Thus,
\begin{equation}\label{equation:upper-bound-on-distance}
    \hat{q}_{j'} - \hat{q}_{j} \geq d_{G}(j, j') - 2\delta \cdot \diam{G(C, E, \succ^k)[\delta], V(\succ^k)}
\end{equation}

Now, consider $j^* \in \arg\max_{j \in V(\succ^k)} \hat{q}_j$.
By Equation \ref{equation:upper-bound-on-distance}, for all $j \in V(\succ^k)$,
\[
    d_{G}(j^*, j) - 2\delta \cdot \diam{G(C, E, \succ^k)[\delta], V(\succ^k)} \leq \hat{q}_j - \hat{q}_{j^*} \leq 0,
\]
In particular, $d_{G}(j^*, j) \leq 2\delta \cdot \diam{G(C, E, \succ^k)[\delta], V(\succ^k)}$.

On the other hand, for all $j \in C \setminus V(\succ^k)$, there exists $i \in [n]$ such that $\metric{e_i}{c_j} \leq r$ since $E$ forms an $r$-cover of $C$.
Let $\ell$ be the favorite candidate of $i$. By definition of $w$,
\[
    w(\ell, j) \leq \metric{e_i}{c_j} - \metric{e_i}{c_\ell} \leq \metric{e_i}{c_j} \leq r
\]
It follows from the triangle inequality and Equation \ref{equation:upper-bound-on-distance} that
\[
    d_{G}(j^*, j) \leq d_{G}(j^*, \ell) + w(\ell, j) \leq 2\delta \cdot \diam{G(C, E, \succ^k)[\delta], V(\succ^k)}+r
\]
Thus, 
\[
     \max_{j \in [m]} d_{G}(j^*, j) \leq r + 2\delta \cdot \diam{G(C, E, \succ^k)[\delta], V(\succ^k)}
\]
\end{proof}



Next, we demonstrate a universal committee $E^*$ such that $E^*$ is an $(\varepsilon/2)$-cover of $\Theta$\footnote{Lemma~\ref{lemma:main-lemma} only requires $E$ to be an $r$-cover of $C$. However, since a universal committee is chosen before the candidates are revealed, to be an $r$-cover for any $C$, the committee needs to be an $r$-cover of $\Theta$.} of size $(8(m-1)/k + 1)(N(\Theta, \norm{\cdot}, \varepsilon/2)-1)$ and that, for all $C$ and inducible $\succ^k$,
\[
    \textstyle \diam{G(C, E^*, \succ^k)\left[\frac{k\varepsilon}{8(m-1)}\right], V(\succ^k)} \leq O(m/k)
\]

\begin{construction}\label{construction:upper-bound}
Let $X \subseteq \Theta$ be an (internal) $(\varepsilon/2)$-cover of $(\Theta, \norm{\cdot})$ of minimal size, and let $H(X)$ denote the undirected graph with vertices in $X$ and an edge between two vertices if and only if the distance between them is at most two times the radius of the cover.
\[
    \textstyle H(X) \coloneqq (X, \{\{x, y\} \in X^2 : 0 < \metric{x}{y} \leq \varepsilon \}).
\]
Let $T$ be a spanning tree of $H(X)$.
For each edge $\{x,y\} \in T$, place an expert at $x+\frac{ki}{8(m-1)}(y-x) \in \Theta$ for $i = 0, 1, \dots, 8(m-1)/k$.
Let $E^* \coloneqq \{e_1, \dots, e_n\}$ denote this placement of experts.
\end{construction}




Note that the spanning tree $T$ exists since $H(X)$ is a connected graph.

\begin{lemma}\label{lemma:H-connected}
Let $X$ be an $(\varepsilon/2)$-cover of $\Theta$.
$H(X) \coloneqq (X, \{\{x, y\} \in X^2 : 0 < \metric{x}{y} \leq \varepsilon \})$ is connected.
\end{lemma}

\begin{proof}
Suppose $H$ is not connected, so there are at least two connected components.
Let $x$ and $y$ minimize $\metric{x}{y}$ over all $x, y \in X$ such that $x$ and $y$ belong to different connected components of $H$. 
Let $C_1 \not= C_2$ denote the connected components containing $x$ and $y$, respectively.
Since $x$ and $y$ belong to different connected components, $\metric{x}{y} > \varepsilon$.
Consider $z = \frac{x + y}{2}$.
Since $\Theta$ is convex, $z \in \Theta$.
Note that $\metric{x}{z} = \metric{y}{z} > \varepsilon/2$, so neither $x$ nor $y$ covers $z$.
But no other element of $X$ could cover $z$ either.
Suppose $w \in X$ covers $z$.
Then, $\metric{w}{z} \leq \varepsilon/2 < \metric{x}{z} = \metric{y}{z}$, so $\metric{x}{y} = \metric{x}{z} + \metric{y}{z} > \metric{x}{z} + \metric{w}{z} \geq \metric{w}{x}$ (similarly, $\metric{x}{y} > \metric{w}{y}$).
If $w \not\in C_1$, then $w$ and $x$ contradict the definition of $x$ and $y$, so $w \in C_1$.
But if $w \not\in C_2$, then $w$ and $y$ contradict the definition of $x$ and $y$.
Thus, $w \in C_1 \cap C_2$, another contradiction.
It follows that no element of $X$ covers $z$, contradicting the fact that $X$ is a cover of $\Theta$.
\end{proof}


We now show that any two candidates who rank first for some expert remain connected after removing $k-1$ other candidates from $G(C, E^*, \succ^k)\left[\frac{k\varepsilon}{8(m-1)}\right]$ (Lemma \ref{lemma:remains-connected}), from which it will follow that there exists a path of length $O(m/k)$ between these two candidates (Lemma \ref{lemma:bound-on-diameter}).

\begin{lemma}\label{lemma:remains-connected}
Let $C = \{c_1, \dots, c_m\}$ denote the locations of $m$ candidates, $E^* = \{e_1, \dots, e_n\}$ denote the locations of the $n$ experts placed according to Construction \ref{construction:upper-bound}, and let $\succ^k = \{\succ_1^k, \dots, \succ_n^k\}$ denote their reported rankings.
For all $J \subseteq [m]$ such that $\abs{J} \leq k-1$ and $j \not= j' \in V(\succ^k) \setminus J$, there exists a $j \to j'$ path in $G(C, E^*, \succ^k)\left[\frac{k\varepsilon}{8(m-1)}\right] - J$.
\end{lemma}

\begin{proof}
Let $J \subseteq [m]$ be an arbitrary subset of at most $k - 1$ candidates, and consider $j \not= j' \in V(\succ^k) \setminus J$.
Since we only removed at most $k-1$ candidates, each expert has at least one candidate from her top $k$ remaining in $V(\succ^k) \setminus J$.
For each expert $i \in [n]$, let $v_i$ denote the candidate at the top of her ranking among $V(\succ^k) \setminus J$.
Since $j, j' \in V(\succ^k) \setminus J$, there exists experts $i, i' \in [n]$ such that $v_i = {\succ_i^k}(1) = j$ and $v_{i'} = {\succ_{i'}^k}(1) = j'$.

Now, recall the spanning tree $T$ from Construction \ref{construction:upper-bound}, and note that for each edge $\{x, y\} \in T$, the distance between any two experts along the line segment between $x$ and $y$ is at most $\frac{k\varepsilon}{8(m-1)}$: for all $\ell = 0, \dots, 8(m-1)/k - 1$,
\begin{align*}
    \textstyle \metric{x + \frac{k(\ell+1)}{8(m-1)}(y-x)}{x + \frac{k\ell}{8(m-1)}(y-x)} 
        &= \norm{\left(x + \frac{k(\ell+1)}{8(m-1)}(y-x)\right)-\left(x + \frac{k\ell}{8(m-1)}(y-x)\right)} \\
        &= \frac{k}{8(m-1)} \norm{y-x} \\
        &\leq \frac{k\varepsilon}{8(m-1)}
\end{align*}
where the last inequality follows from the fact that $\{x,y\}$ is an edge in $T$, so $\metric{x}{y} \leq \varepsilon$.

Thus, there exists a sequence of experts $i = i_1, i_2, \dots, i_S = i'$ such that $\metric{e_{i_s}}{e_{i_{s+1}}} \leq \frac{k\varepsilon}{8(m-1)}$ for all $s \in [S - 1]$, namely, the sequence of experts encountered as we walk along the path from $i$ to $i'$ along the edges of $T$ (or rather, the line segments in $\Theta$ that correspond to the edges of $T$).
Construct a $j \to j'$ path as follows. 
Let $j_1 := j$.
Let $s \in [S-1]$ denote the index of the last expert in the sequence such that $v_{i_{s}} = j_\ell$ (so $v_{i_{s+1}} \not= j_1, \dots, j_\ell$). 
Define $j_{\ell+1} \coloneqq v_{i_{s+1}}$.
Since $\metric{i_{s}}{i_{s+1}} \leq \frac{k\varepsilon}{8(m-1)}$, $(j_\ell, j_{\ell + 1})$ is an edge in $G(C, E^*, \succ^k)\left[\frac{k\varepsilon}{8(m-1)}\right] - J$.
Terminate this process when $j_{\ell + 1} = j'$ (which will occur because $S < +\infty$ and $v_{i_S} = v_{i'} = j'$).
Note that the returned path is a $j \to j'$ path in $G(C, E^*, \succ^k)\left[\frac{k\varepsilon}{8(m-1)}\right] - J$.
\end{proof}

\begin{lemma}\label{lemma:bound-on-diameter}
Let $C = \{c_1, \dots, c_m\}$ denote the locations of $m$ candidates, $E^* = \{e_1, \dots, e_n\}$ denote the locations of the $n$ experts placed according to Construction \ref{construction:upper-bound}, and let $\succ^k = \{\succ_1^k, \dots, \succ_n^k\}$ denote their reported rankings.
For all $j \not= j' \in V(\succ^k)$, there exists a $j \to j'$ path in $G(C, E^*, \succ^k)\left[\frac{k\varepsilon}{8(m-1)}\right]$ with at most $2(m-1)/k$ edges.
\end{lemma}

\begin{proof}
The lemma holds trivially for any $j \not= j' \in V(\succ^k)$ that are adjacent in $G(C, E^*, \succ^k)\left[\frac{k\varepsilon}{8(m-1)}\right]$.
Thus, suppose $j \not= j' \in V(\succ^k)$ are not adjacent.
The minimum vertex cut $J$ that disconnects $j'$ from $j$ consists of at least $k$ candidates.
Otherwise, there is no $j \to j'$ path in $G(C, E^*, \succ^k)\left[\frac{k\varepsilon}{8(m-1)}\right] - J$, yet $\abs{J} \leq k-1$, contradicting Lemma \ref{lemma:remains-connected}.
Thus, by Menger's Theorem, there are at least $k$ vertex-independent $j \to j'$ paths in $G(C, E^*, \succ^k)\left[\frac{k\varepsilon}{8(m-1)}\right]$.
It follows that the $j \to j'$ path in $G(C, E^*, \succ^k)\left[\frac{k\varepsilon}{8(m-1)}\right]$ with the fewest number of edges contains at most $(m-2)/k + 1 \leq 2(m-1)/k$ edges.
\end{proof}

We now combine Lemmas \ref{lemma:main-lemma} and \ref{lemma:bound-on-diameter} to derive Lemma \ref{lemma:E^*-low-regret} as follows.
By Lemma \ref{lemma:bound-on-diameter}, there exists a path in $G(C, E^*, \succ^k)\left[\frac{k\varepsilon}{8(m-1)}\right]$ of length $O(m/k)$ between any two candidates who rank first for some expert.
Moreover, $E^*$ is an $(\varepsilon/2)$-cover of $\Theta$. 
Thus, by Lemma \ref{lemma:main-lemma}, the worst-case regret of choosing a candidate with the greatest proxy quality among those who rank first for some expert is at most $O(\varepsilon)$.

\begin{lemma}\label{lemma:E^*-low-regret}
Let $E^*$ be constructed according to Construction \ref{construction:upper-bound}.
$E^*$ is a universal committee of size $(8(m-1)/k + 1) (N\left(\Theta, \norm{\cdot}, \varepsilon/2\right) - 1)$.
Moreover, for any set $C = \{c_1,\ldots, c_m\}$ of $m$ candidate locations and any set $\succ^k = \{\succ_1^k, \dots, \succ_n^k\}$ of inducible rankings given $C$ and $E^*$, if $\hat{q} \in \RR^m$ is consistent with $\succ^k$, then, for all $j^* \in \arg\max_{j \in V(\succ^k)} \hat{q}_j$,
\[
    \max_{j \in [m]} d_G(j^*, j) \leq \varepsilon
\]
\end{lemma}

\begin{proof}
Let $X \subseteq \Theta$ denote the $(\varepsilon/2)$-cover underlying the construction of $E^*$.
By Lemma \ref{lemma:H-connected}, $H(X)$ is connected, so the spanning tree $T$ underlying the construction of $E^*$ exists.
Moreover, $E^* \subseteq \Theta$ by the convexity of $\Theta$.
Thus, $E^*$ is a well-defined universal committee.
Since $X$ is an $(\varepsilon/2)$-cover of $\Theta$ of minimum size, $\abs{X} = N(\Theta, \norm{\cdot}, \varepsilon/2)$ and $T$ contains $N(\Theta, \norm{\cdot}, \varepsilon/2) - 1$ edges.
Since $E^*$ contains $8(m-1)/k + 1$ experts per edge in $T$, there are $(8(m-1)/k + 1)(N(\Theta, \norm{\cdot}, \varepsilon/2)-1)$ experts in total.

Now, let $C:= \{c_1,\ldots, c_m\}$ denote the locations of $m$ arbitrary candidates and $\succ^k \coloneqq \{\succ_1^k, \dots, \succ_n^k\}$ denote a set of inducible rankings.
Let $\hat{q} \in \RR^m$ be any quality vector consistent with $\succ^k$, and let $j^* \in \arg\max_{j \in V(\succ^k)} \hat{q}_j$.
Since $E^*$ is an $(\varepsilon/2)$-cover of $\Theta$ (the experts at $X$ already form an $(\varepsilon/2)$-cover),
\begin{align*}
\max_{j \in [m]} d_{G}(j^*, j) 
        &\leq \varepsilon/2 + 2\left(\frac{k\varepsilon}{8(m-1)}\right) \textstyle \diam{G(C, E^*, \succ^k)\left[\frac{k\varepsilon}{8(m-1)}\right], V(\succ^k)} \tag{Lemma \ref{lemma:main-lemma} with $r \coloneqq \varepsilon/2, \delta \coloneqq \frac{k\varepsilon}{8(m-1)}$} \\
        &\leq \varepsilon/2 + 2\left(\frac{k\varepsilon}{8(m-1)}\right)\left(\frac{2(m-1)}{k}\right) \tag{Lemma \ref{lemma:bound-on-diameter}} \\
        &= \varepsilon
\end{align*}

\end{proof}

With Lemma \ref{lemma:E^*-low-regret} in hand, Theorem \ref{theorem:general-upper-bound} easily follows.

\begin{proof}[Proof of Theorem \ref{theorem:general-upper-bound}]
Let $f$ denote the generalized minimal regret voting rule, $E^* = \{e_1, \dots, e_n\}$ denote the locations of the $n$ experts placed according to Construction \ref{construction:upper-bound}, $C:= \{c_1,\ldots, c_m\}$ denote the locations of $m$ arbitrary candidates, and $\succ^k \coloneqq \{\succ_1^k, \dots, \succ_n^k\}$ denote a set of inducible rankings.
Let $\hat{q} \in \RR^m$ be any quality vector consistent with $\succ^k$, and let $j^* \in \arg\max_{j \in V(\succ^k)} \hat{q}_j$.
Note that
\[
    \max_{j \in [m]} d_{G}(f(C, E^*, \succ^k), j) \leq \max_{j \in [m]} d_{G}(j^*, j) \leq \varepsilon
\]
where first inequality follows from the fact that $f(C, E^*, \succ_i^k) = \arg\min_{j \in C} \max_{j' \in C} d_{G}(j, j')$ and the second inequality follows from Lemma \ref{lemma:E^*-low-regret}.
Thus, by Corollary \ref{corollary:reduction}, the regret of selecting the minimal regret candidate $f(C, E^*, \succ^k)$ is at most $\varepsilon$.
\end{proof}


In line with our discussion of Lemma \ref{lemma:main-lemma}, we derive as a corollary of Lemma \ref{lemma:E^*-low-regret} an alternative voting rule that is more efficient than the generalized minimal regret voting rule when $G(C, E^*, \succ^k)$ is sparse.
The alternative voting rule solves a single-source shortest path problem to compute a proxy quality vector that is consistent with the partial preferences returned by the experts and returns the candidate with the greatest proxy quality among those who rank first for some candidate.

\begin{corollary}[Corollary to Lemma \ref{lemma:E^*-low-regret}]\label{corollary:new-voting-rule}
Let $E^* = \{e_1,\dots,e_n\}$ be constructed according to Construction \ref{construction:upper-bound}.
There exists an algorithm that takes candidate locations $C = \{c_1,\ldots, c_m\}$ and the reported rankings $\succ^k = \{\succ_1^k, \dots, \succ_n^k\}$ of the experts as inputs and outputs a candidate with a regret of at most $\varepsilon$ in $O(nmk + m \cdot \abs{E(G(C, E^*, \succ^k))})$ time where $E(G)$ denotes the edge set of a graph $G$.
\end{corollary}


\section{General Lower Bound}\label{section:general-lower-bound}

\begin{theorem}\label{theorem:general-lower-bound}
Let $(V, \norm{\cdot})$ be a normed vector space, $\varepsilon > 0$, and $\Theta \subseteq V$ be convex with $\mathrm{diam}(\Theta) \geq 12\varepsilon$.
If each expert returns a ranking of her top $k$ candidates, then any universal committee that guarantees a regret of at most $\varepsilon$ for any set of $m \geq k(N(B(0,1), \norm{\cdot}, 1/4) + 1)$ candidates in $\Theta$ under some deterministic voting rule requires size at least
\[
    M(\Theta, \norm{\cdot}, 24\varepsilon)\left(\frac{m}{k(N(B(0, 1), \norm{\cdot}, 1/4) + 1)} - 1\right)
\]
\end{theorem}

To show the lower bound, we will show that for any universal committee that guarantees $\varepsilon$ regret for any $m$ candidates and any $(24\varepsilon)$-packing of $\Theta$, each ball (of radius $12\varepsilon$) of the packing must contain $\Omega(m/(k \cdot N(B(0,1), \norm{\cdot}, 1/4)))$ experts. 
Otherwise, there exists a way of placing $m$ candidates that admits two quality vectors $q^1$ and $q^2$ with the following properties: (1) the two quality vectors induce the same ranking from each expert, so the planner cannot distinguish between them and will choose the same candidate under both and (2) regardless of the candidate that is chosen, the regret of choosing her is at least $\varepsilon$ under one of the two quality vectors.

To give an idea of the construction, suppose there is a ball with too few experts.
We will place $k$ candidates at the center of the ball and for each expert in the ball, we will place $k$ candidates at her location.
Then, for each group of $k$ candidates just placed, we will place candidates on the boundary of the ball.
In $q^1$, the candidates at the center of the ball will have qualities of $0$, and the qualities of candidates will increase linearly with their distance from the center.
In $q^2$, the candidates at the center of the ball will have qualities that are $\Omega(\varepsilon)$, and the qualities of candidates will decrease linearly with their distance from the center.
Thus, if the same candidate is chosen under the two quality vectors, then a high quality under one vector implies a low quality under the other, and vice versa.

We will see that the top $k$ candidates for each expert inside the ball will consist of the $k$ candidates at her location under both $q^1$ and $q^2$.
Moreover, we will see that the top $k$ candidates for each expert outside of the ball will consist of some $k$ candidates outside of the ball under both quality vectors.
In other words, the candidates on the boundary will protect those inside the ball from being in the top $k$ for experts outside of the ball.
This is what allows us to move between $q^1$ and $q^2$ without changing some expert's top $k$ ranking.

We remark that the idea to partition the underlying space and then argue what happens when one element of the partition contains too few experts and the idea to construct two quality vectors that yield indistinguishable votes can be found in \citet{Feldman_Mansour_Nisan_Oren_Tennenholtz_2020}.
Our contribution is a generalization of these ideas 
to the case when some but still too few experts lie within an element of the partition: \citet{Feldman_Mansour_Nisan_Oren_Tennenholtz_2020} only consider the case when there exists an element of the partition with no expert. 
In particular, 
our realization that an $O(r/2)$-cover of the boundary of a ball of radius $r$ centered at a candidate is precisely what prevents any expert outside of the ball from voting for the candidate at the center resolve critical issues that more straightforward extensions of \citet{Feldman_Mansour_Nisan_Oren_Tennenholtz_2020}'s ideas cannot address.

\subsection{Proof of Theorem \ref{theorem:general-lower-bound}}


Before outlining the construction behind our lower bound result, we introduce some notation for the sets that play recurring roles.

\begin{definition}
Let $\partial S \coloneqq \closure{S} \setminus \relint{S}$ denote the relative boundary of $S$.
\end{definition}

\begin{definition}
Let $(V, \norm{\cdot})$ be a normed vector space and $\Theta \subseteq V$.
For all $x \in \RR^d$ and $r > 0$, let $B_\Theta(x, r) \coloneqq \Theta \cap B(x, r)$ denote the intersection between $\Theta$ and the ball of radius $r$ centered at $x$.
\end{definition}

\begin{definition}
Let $L(x, y) \coloneqq \{x + \alpha (y-x) : \alpha \in [0, 1]\}$ denote the line segment between $x, y \in \RR^d$ and $S[A, B] \coloneqq \{S \cap L(x,y) : x \in A, y \in B\}$ denote the intersection between $S \subseteq \RR^d$ and all line segments between $A, B \subseteq \RR^d$.
If $A = \{a\}$ is a singleton, then we write $S[a, B]$ instead of $S[A, B]$ (the same applies for $B$ as well if $B$ were a singleton).
\end{definition}

\begin{figure}
    \includegraphics[width=\linewidth]{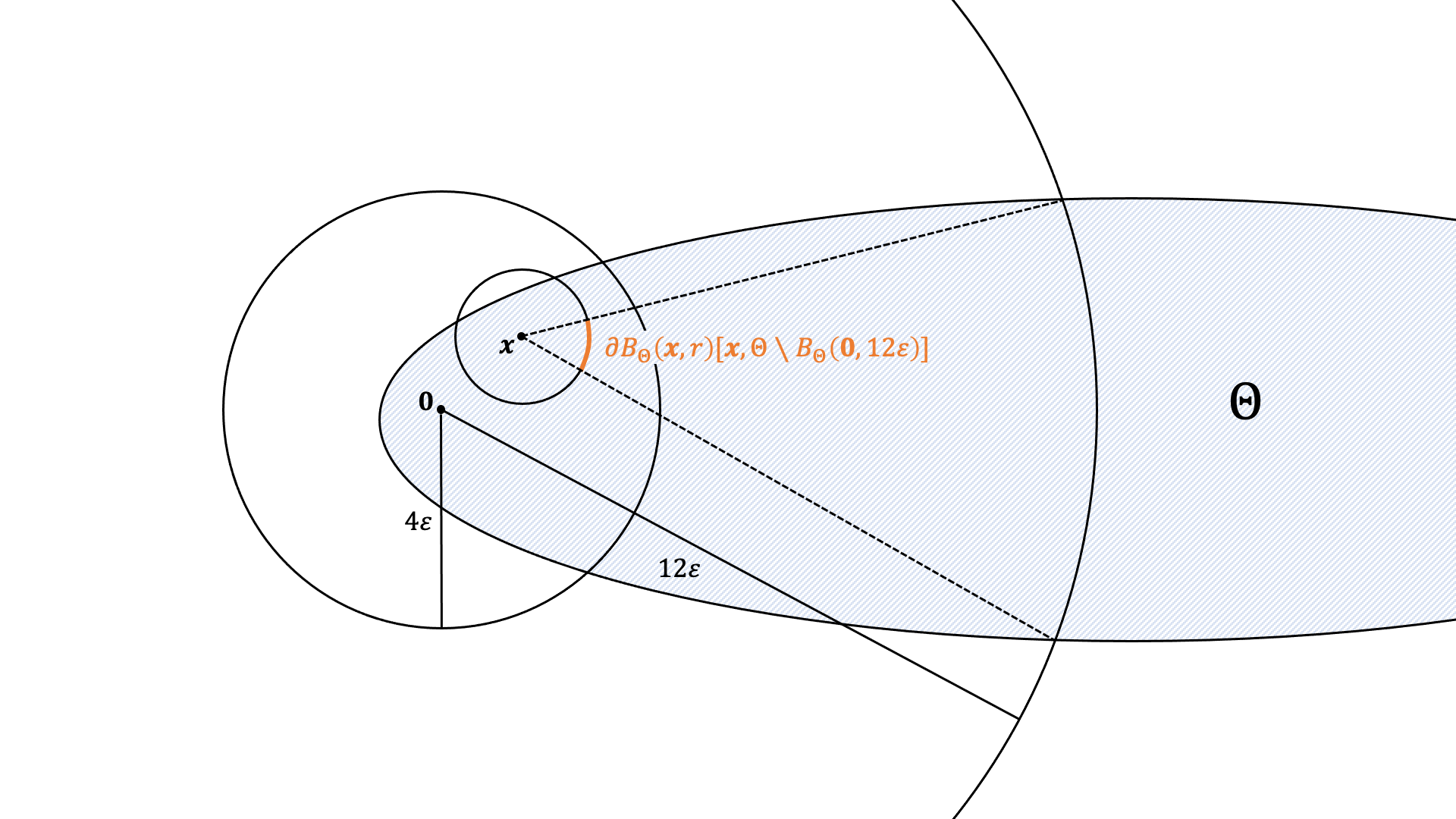}
    \caption{The set $\partial B_\Theta(x, r)[x, \Theta \setminus B_\Theta(0, 12 \varepsilon)]$}
\end{figure}

Let $E \coloneqq \{e_1, \dots, e_n\} \subseteq \Theta$ denote an arbitrary universal committee that guarantees a regret of at most $\varepsilon > 0$ for any set of $m$ candidates in $\Theta$.
Let $X$ denote an $(24\varepsilon)$-packing of $\Theta$ of maximum size (so $\abs{X} = M(\Theta, \norm{\cdot}, 24\varepsilon)$).
We will show that for all $x \in X$, 
\begin{equation}\label{equation:desired-lower-bound-for-each-ball}
    \abs{E \cap B_\Theta(x, 12\varepsilon)} \geq \frac{m}{k(N(B(0, 1), \norm{\cdot}, 1/4) + 1)} - 1.
\end{equation}
Since $X$ is a $(24\varepsilon)$-packing, $B(x, 12\varepsilon) \cap B(y, 12\varepsilon) = \varnothing$ for all $x \not= y \in X$.
Thus, 
\[
    \abs{E} \geq M(\Theta, \norm{\cdot}, 24\varepsilon)\left(\frac{m}{k(N(B(0, 1), \norm{\cdot}, 1/4) + 1)} - 1\right).
\]

Now, suppose there exists $x \in X$ for which Equation \ref{equation:desired-lower-bound-for-each-ball} does not hold.
Suppose without loss of generality that $x = 0 \in \Theta$.
Construct a set of at most $m$ candidates as follows.
\begin{enumerate}
    \item Place $k$ candidates at $0 \in \Theta$ with $q^1_j = 2\varepsilon$ and $q^2_j = 0$.
    \item Let $W_0$ be a $(2\varepsilon)$-cover of $\partial B_\Theta(0, 4\varepsilon)[0, \Theta \setminus B_\Theta(0, 12\varepsilon)]$ of minimal size (so $\abs{W_0} = N(\partial B_\Theta(0, 4\varepsilon)[0, \Theta \setminus B_\Theta(0, 12\varepsilon)], \norm{\cdot}, 2\varepsilon)$).
    For each $w \in W_0$, place $k$ candidates at $w$ with $q^1_j = 0$ and $q^2_j = 2\varepsilon$.
    Note that these candidates lie on $\partial B_\Theta(0, 4\varepsilon) \subseteq \Theta$.
    We will see that these candidates prevent any expert in $\Theta \setminus B_\Theta(0, 12\varepsilon)$ from voting for the candidate at $0$.
    \item For each expert $i \in [n]$ such that $e_i \in \relint{B_\Theta(0, 4\varepsilon)} \subseteq \Theta$,
    \begin{enumerate}
        \item Place $k$ candidates at $e_i$ with $q^1_j = 2\varepsilon - \frac{1}{2} \norm{e_i}$ and $q^2_j = \frac{1}{2} \norm{e_i}$.
        \item Let $W_i$ be a $(2\varepsilon - \frac{1}{2}\norm{e_i})$-cover of $\partial B_\Theta(e_i, 4\varepsilon - \norm{e_i})[e_i, \Theta \setminus B_\Theta(0, 12\varepsilon)]$ of minimal size (so $\abs{W_i} = N(\partial B_\Theta(e_i, 4\varepsilon - \norm{e_i})[e_i, \Theta \setminus B_\Theta(0, 12\varepsilon)], \norm{\cdot}, 2\varepsilon - \frac{1}{2} \norm{e_i})$).
        For each $w \in W_i$, let $\lambda \geq 0$ be such that $\norm{e_i + \lambda (w - e_i)} = 4\varepsilon$ and place $k$ candidates at $e_i + \lambda (w - e_i)$ with $q^1_j = 0$ and $q^2_j = 2\varepsilon$.
        By Lemma \ref{lemma:candidates-lie-in-Theta}, these candidates lie on $\partial B_\Theta(0, 4\varepsilon) \subseteq \Theta$.
        We will see that these candidates prevent any expert in $\Theta \setminus B_\Theta(0, 12\varepsilon)$ from voting for the candidate at $e_i$.
    \end{enumerate}
    \item For each expert $i$ such that $e_i \in B_\Theta(0, 12\varepsilon) \setminus \relint{B_\Theta(0, 4\varepsilon)} \subseteq \Theta$, place $k$ candidates at $e_i$ with $q^1_j = 0$ and $q^2_j = 2\varepsilon$.
\end{enumerate}

\begin{figure}
    \includegraphics[width=\linewidth]{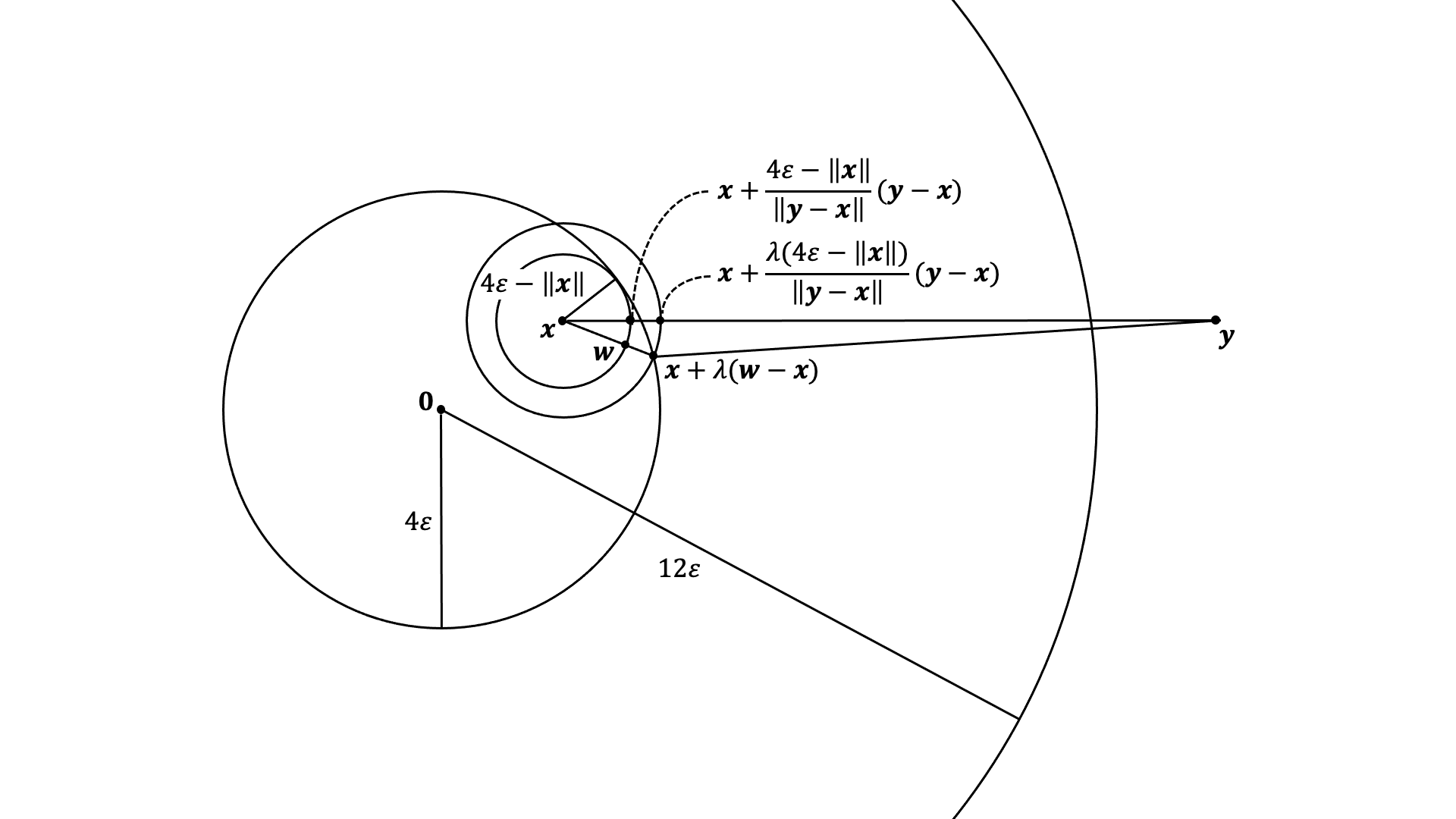}
    \caption{Picture proof of Lemma \ref{lemma:no-outside-expert-votes-for-an-inside-candidate}}
\end{figure}

\begin{lemma}\label{lemma:bound-on-W-covers}
For all $x \in \relint{B_\Theta(0, 4 \varepsilon)}$ and $r > 0$ such that $B_\Theta(x, r) \subseteq B_\Theta(0, 4\varepsilon)$,
\[
    N(\partial B_\Theta(x, r)[x, \Theta \setminus B_\Theta(0, 12 \varepsilon)], \norm{\cdot}, r/2) \leq N(B(0, 1), \norm{\cdot}, 1/4)
\]
\end{lemma}

\begin{lemma}\label{lemma:candidates-lie-in-Theta}
Let $x \in \relint{B_\Theta(0, 4\varepsilon)}$ and $W$ denote a $(2\varepsilon - \frac{1}{2} \norm{x})$-cover of $\partial B_\Theta(x, 4\varepsilon - \norm{x})[x, \Theta \setminus B_\Theta(0, 12 \varepsilon)]$.
For all $w \in W$, $x + \lambda (w - x) \in \Theta$ where $\lambda$ solves $\norm{x + \lambda (w - x)} = 4\varepsilon$ subject to $\lambda \geq 0$.
\end{lemma}

Note that our construction places
\begin{align*}
k(\abs{W_0} + 1) & {} + \sum_{i : e_i \in \relint{B_\Theta(0, 4\varepsilon)}} k(\abs{W_i} + 1) + \sum_{i : e_i \in B_\Theta(0, 12\varepsilon) \setminus \relint{B_\Theta(0, 4\varepsilon)}} k \\
    &\leq k(N(B(0, 1), \norm{\cdot}, 1/4) + 1) + \abs{E \cap B_\Theta(0, 12\varepsilon)} \cdot k(N(B(0, 1), \norm{\cdot}, 1/4) + 1) \tag{Lemma \ref{lemma:bound-on-W-covers}} \\
    &< m \tag{$\abs{E \cap B_\Theta(x, 12\varepsilon)} < \frac{m}{k(N(B(0, 1), \norm{\cdot}, 1/4) + 1)} - 1$} 
\end{align*}
candidates.
Moreover, by Lemma \ref{lemma:candidates-lie-in-Theta}, each candidate lies in $\Theta$, so our construction defines a valid set of candidates.

\subsection*{No expert in $\Theta \setminus B_\Theta(0, 12\varepsilon)$ will vote for a candidate in $\relint{B_\Theta(0, 4\varepsilon)}$}

\begin{lemma}\label{lemma:no-outside-expert-votes-for-an-inside-candidate}
Let $x \in \relint{B_\Theta(0, 4\varepsilon)}$, $y \in \Theta \setminus B_\Theta(0, 12\varepsilon)$, and $W$ denote a $(2\varepsilon - \frac{1}{2} \norm{x})$-cover of $\partial B_\Theta(x, 4\varepsilon - \norm{x})[x, \Theta \setminus B_\Theta(0, 12 \varepsilon)]$.
There exists $w \in W$ such that
\[
    \norm{y - \left(x + \lambda (w - x) \right)} \leq \norm{y - x} - \left(2\varepsilon - \frac{1}{2} \norm{x}\right) 
\]
where $\lambda$ solves $\norm{x + \lambda (w - x)} = 4\varepsilon$ subject to $\lambda \geq 0$.
\end{lemma}

\begin{proof}
Since $\norm{x} < 4\varepsilon$ and $\norm{y - x} \geq \norm{y} - \norm{x} > 8 \varepsilon$, $\frac{4\varepsilon - \norm{x}}{\norm{y - x}} \in (0, 1/2)$.
Thus, $x + \frac{4\varepsilon - \norm{x}}{\norm{y - x}} (y - x) \in L(x, y) \subseteq \Theta$ since $\Theta$ is convex and $x, y \in \Theta$.
Moreover,
\[
    \norm{x - \left(x + \frac{4\varepsilon - \norm{x}}{\norm{y - x}} (y - x)\right)} = 4\varepsilon - \norm{x}
\]
so $x + \frac{\lambda \left(4\varepsilon - \norm{x}\right)}{\norm{y - x}} (y - x) \in \partial B_\Theta(x, 4\varepsilon - \norm{x})[x, \Theta \setminus B_\Theta(0, 12 \varepsilon)]$.
Since $W$ is a $(2\varepsilon - \frac{1}{2} \norm{x})$-cover of $\partial B_\Theta(x, 4\varepsilon - \norm{x})[x, \Theta \setminus B_\Theta(0, 12 \varepsilon)]$, there exists $w \in W$ such that $\norm{\left[x + \frac{4\varepsilon - \norm{x}}{\norm{y - x}} (y-x)\right] - w} \leq 2\varepsilon - \frac{1}{2} \norm{x}$.
Note that
\[
     1 = \frac{4\varepsilon - \norm{x}}{\norm{w - x}} \leq \lambda \leq \frac{4\varepsilon + \norm{x}}{\norm{w - x}}= \frac{4\varepsilon + \norm{x}}{4\varepsilon - \norm{x}}
\]
where the inequalities follow from the fact that $\lambda$ solves $\norm{x + \lambda \left(4\varepsilon - \norm{x}\right) w} = 4\varepsilon$ subject to $\lambda \geq 0$ and the equalities follow from the fact that $w \in W \subseteq \partial B_\Theta(x, 4\varepsilon - \norm{x})$.
Moreover, since $\norm{y - x} \geq \norm{y} - \norm{x} > 8 \varepsilon$, $\frac{\lambda (4\varepsilon - \norm{x})}{\norm{y - x}} \leq \frac{4\varepsilon + \norm{x}}{\norm{y - x}} < 1$.
Thus,
\begin{align*}
    \norm{y - \left(x + \lambda (w - x) \right)}
        \leq {} & \norm{y - \left[x + \frac{\lambda \left(4\varepsilon - \norm{x}\right)}{\norm{y - x}} (y - x)\right]} \\
            & + \norm{\left[x + \frac{\lambda \left(4\varepsilon - \norm{x}\right)}{\norm{y - x}} (y - x)\right] - (x + \lambda (w - x))} \\
        = {} & \norm{y - x}\left(1 - \frac{\lambda \left(4\varepsilon - \norm{x}\right)}{\norm{y - x}}\right) + \lambda \norm{\left[x + \frac{4\varepsilon - \norm{x}}{\norm{y - x}} (y-x)\right] - w} \tag{$\frac{\lambda (4\varepsilon - \norm{x})}{\norm{y - x}} < 1$} \\
        \leq {} & \norm{y - x} - \lambda \left(4\varepsilon - \norm{x}\right) + \frac{\lambda}{2}(4\varepsilon - \norm{x}) \tag{$\norm{\left[x + \frac{4\varepsilon - \norm{x}}{\norm{y - x}} (y-x)\right] - w} \leq 2\varepsilon - \frac{1}{2} \norm{x}$} \\
        \leq {} & \norm{y - x} - \left(2\varepsilon - \frac{1}{2} \norm{x}\right) \tag{$\lambda \geq 1$}
\end{align*}
\end{proof}

By Lemma \ref{lemma:no-outside-expert-votes-for-an-inside-candidate}, for all $x \in \{0\} \cup (E \cap \relint{B_\Theta(0, 4\varepsilon)})$ and $y \in \Theta \setminus B_\Theta(0, 12\varepsilon)$, there exists $w \in W_0 \cup \bigcup_{i \in [n] : e_i \in \relint{B_\Theta(0, 4\varepsilon)}} W_i$ such that, letting $j$ denote the candidate placed at $x$ and $j'$ denote the candidate placed at $x + \lambda (w - x)$,
\begin{align*}
    q^1_{j'} - \norm{y - (x + \lambda (w - x))}
        &= - \norm{y - (x + \lambda (w - x))} \tag{$q^1_{j'} = 0$} \\
        &\geq \left(2\varepsilon - \frac{1}{2} \norm{x}\right) - \norm{y - x} \tag{Lemma \ref{lemma:no-outside-expert-votes-for-an-inside-candidate}} \\
        &= q^1_j - \norm{y - x} \tag{$q^1_{j} = 2\varepsilon - \frac{1}{2} \norm{x}$} \\
    q^2_{j'} - \norm{y - (x + \lambda (w - x))}
        &= 2 \varepsilon - \norm{y - (x + \lambda (w - x))} \tag{$q^1_{j'} = 2\varepsilon$} \\
        &\geq 2 \varepsilon + \left(2\varepsilon - \frac{1}{2} \norm{x}\right) - \norm{y - x} \tag{Lemma \ref{lemma:no-outside-expert-votes-for-an-inside-candidate}} \\
        &> \frac{1}{2} \norm{x} - \norm{y - x} \tag{$\norm{x} < 4\varepsilon$} \\
        &= q^2_j - \norm{y - x} \tag{$q^2_{j} = \frac{1}{2} \norm{x}$} 
\end{align*}
In other words, no expert in $\Theta \setminus B_\Theta(0, 12\varepsilon)$ will rank any of the $k$ candidates at $x$ in her top $k$ candidates.
Thus, the top $k$ candidates for each of these experts will consist of candidates either on the boundary of $B_\Theta(0, 4\varepsilon)$ or in $B_\Theta(0, 12\varepsilon) \setminus B_\Theta(0, 4\varepsilon)$.\footnote{Note that one of the inequalities is not strict, so some candidate in $\relint{B_\Theta(0, 4\varepsilon)}$ may tie with an outside candidate and end up among the top $k$ candidates for some expert outside of $B_\Theta(0, 12\varepsilon)$. One can address this issue by decreasing the radius of the cover in Lemma \ref{lemma:no-outside-expert-votes-for-an-inside-candidate} by an arbitrarily small $\delta$ without qualitatively changing our results.} 

\subsection*{Each expert in $\relint{B_\Theta(0, 4\varepsilon)}$ votes for the candidate at her location}

\begin{lemma}\label{lemma:inside experts do not prefer outside candidates}
If $x \in \relint{B_\Theta(0, 4\varepsilon)}$ and $y \in \Theta \setminus \relint{B_\Theta(0, 4\varepsilon)}$, then
\begin{enumerate}
    \item $2\varepsilon - \frac{1}{2} \norm{x} > -\norm{y - x}$
    \item $\frac{1}{2} \norm{x} > 2\varepsilon - \norm{y - x}$
\end{enumerate}
\end{lemma}

\begin{proof}
The first part of the lemma follows readily from the fact that $\norm{x} < 4\varepsilon$.
To see the second part of the lemma, note that
\[
    2\varepsilon - \norm{y - x} \leq 2\varepsilon - \norm{y} + \norm{x} \leq \norm{x} - 2\varepsilon < \frac{1}{2} \norm{x}
\]
where the first inequality follows from the triangle inequality, the second from the fact that $\norm{y} \geq 4\varepsilon$, and third follows from the fact that $\norm{x} < 4\varepsilon$.
\end{proof}

\begin{lemma}\label{lemma:the best inside candidate}
For $x\neq x' \in \relint{B_\Theta(0, 4\varepsilon)}$,
\begin{enumerate}
    \item $\frac{1}{2} \norm{x} {>} \frac{1}{2} \norm{x'} - \norm{x - x'}$
    \item $2\varepsilon - \frac{1}{2} \norm{x} {>}\left(2\varepsilon - \frac{1}{2} \norm{x'}\right) - \norm{x - x'}$
\end{enumerate}
\end{lemma}

\begin{proof}
To see the first part of the lemma, note that by the triangle inequality,
\[
    \frac{1}{2} (\norm{x'} - \norm{x}) \leq \frac{1}{2} \norm{x - x'} < \norm{x - x'} 
\]
By symmetry, we also have
\[
    \frac{1}{2}\norm{x'} > \frac{1}{2} \norm{x} - \norm{x - x'}
\]
Add $2\varepsilon$ to both sides and rearrange to get the second part of the lemma.
\end{proof}

By Lemma~\ref{lemma:inside experts do not prefer outside candidates}, each expert in $\relint{B_\Theta(0, 4\varepsilon)}$ prefers the $k$ candidates at her location over any candidate in $\Theta \setminus \relint{B_\Theta(0, 4\varepsilon)}$.
By Lemma~\ref{lemma:the best inside candidate}, she prefers the $k$ candidates at her location over any other candidate in $\relint{B_\Theta(0, 4\varepsilon)}$.
Thus, her top $k$ candidates consist of the $k$ candidates at her location under both quality vectors.

\subsection*{Each expert in $B_\Theta(0, 12\varepsilon) \setminus \relint{B_\Theta(0, 4\varepsilon)}$ votes for the candidate at her location}

Consider two experts $i, i' \in [n]$ located at $e_i \not= e_{i'} \in B_\Theta(0, 12\varepsilon) \setminus \relint{B_\Theta(0, 4\varepsilon)}$.
By construction, the candidates at $e_i$ and $e_{i'}$ have the same quality (either $0$ or $2\varepsilon$).
For any quality $q$, it is easy to see that
\[
    q - \norm{e_i - e_i} > q - \norm{e_i - e_{i'}}
\]
so the expert at $e_i$ will always prefer the $k$ candidates at her location over a candidate at another location in $B_\Theta(0, 12\varepsilon) \setminus \relint{B_\Theta(0, 4\varepsilon)}$.

\begin{lemma}\label{lemma:in-between-expert-prefers-candidate-at-her-location}
If $x \in \relint{B_\Theta(0, 4\varepsilon)}, x' \in B_\Theta(0, 12\varepsilon) \setminus \relint{B_\Theta(0, 4\varepsilon)}$, then
\begin{enumerate}
    \item $0 > \left(2\varepsilon - \frac{1}{2} \norm{x}\right) - \norm{x - x'}$ 
    \item $2\varepsilon > \frac{1}{2} \norm{x} - \norm{x - x'}$
\end{enumerate}
\end{lemma}

\begin{proof}
To see the first part of the lemma, note that
\[
    \norm{x - x'} > \frac{1}{2} \norm{x - x'} \geq \frac{1}{2} \left(\norm{x'} - \norm{x}\right) \geq 2\varepsilon - \frac{1}{2} \norm{x} 
\]
where the second inequality follows from the triangle inequality and the third from the fact that $\norm{x'} \geq 4\varepsilon$.
The second part of the lemma follows readily from the fact that $\norm{x} < 4\varepsilon$.
\end{proof}

Lemma \ref{lemma:in-between-expert-prefers-candidate-at-her-location} essentially says that the expert at $x'$ prefers the $k$ candidates at her location over any candidate located in $\relint{B_\Theta(0, 4\varepsilon)}$.
Since this expert prefers the $k$ candidates at her location over all other candidates in $B_\Theta(0, 12\varepsilon) \setminus \relint{B_\Theta(0, 4\varepsilon)}$ as well, it follows that her top $k$ consists of the $k$ candidates at her location.

\subsection*{Indistinguishable votes}

By the previous sections, it is clear that each expert in $B_\Theta(0, 12\varepsilon)$ votes for the $k$ candidates at her location under both $q^1$ and $q^2$.
Moreover, each expert in $\Theta \setminus B_\Theta(0, 12\varepsilon)$ votes for candidates in $\Theta \setminus \relint{B_\Theta(0, 4\varepsilon)}$.
But these candidates have the same qualities as each other under $q^1$ and $q^2$ separately.
Thus, the votes of experts in $\Theta \setminus B_\Theta(0, 12\varepsilon)$ would not change when the quality changes from $q^1$ to $q^2$ and vice versa.
Thus, the two quality vectors induce indistinguishable votes, and the regret is at least $\max\{\frac{1}{2} \norm{x}, 2\varepsilon - \frac{1}{2} \norm{x}\} \geq \varepsilon$ (where $\norm{x} \leq 4\varepsilon$).

\section{Multi-Winner and Multi-Round Voting}\label{section:extensions}

In this section, we consider two settings to which our ideas easily extend.

\subsection{Multi-Winner Voting} 

In the multi-winner setting, we wish design a voting rule $f$ and committee $E$ of minimal size to hire $\ell \leq k$ candidates such that our cumulative regret is at most $\varepsilon$.
That is, we wish to select $f$---a function from the product space of candidate locations, expert locations, and top $k$ rankings to the space of all sets of $\ell$ candidates---and $E$ such that for all sets $C$ of $m$ candidate locations and all quality vectors $q \in \RR^m$,
\[
    \max_{S' \subseteq [m] : \abs{S'} = \ell} \sum_{j \in S'} q_j - \sum_{j \in f(C, E, \succ^k)} q_j \leq \varepsilon
\]
The ideas behind Theorems \ref{theorem:general-upper-bound} and \ref{theorem:general-lower-bound} readily extend to this setting to yield the following result.

\begin{theorem}\label{theorem:multi-winner}
Let $(V, \norm{\cdot})$ be a normed vector space, $B_1$ be a unit ball with respect to $\norm{\cdot}$, and $\Theta \subseteq V$ be convex with $\mathrm{diam}(\Theta) \geq 12\varepsilon > 0$.
Let $\mathcal{F}$ denote the set of deterministic voting rules that take as inputs candidate locations, expert locations, and top $k$ rankings and output a set of $\ell$ candidates.
Given $f \in \mathcal{F}$, let $n_f$ denote the size of the smallest universal committee that guarantees a cumulative regret of at most $\varepsilon$ for any set of $m \geq k(N(B_1, 1/4) + 1)$ candidates in $\Theta$.
\[
    \textstyle M(\Theta, 24\varepsilon/\ell)\left(\frac{m}{k(N(B_1, 1/4) + 1)} - 1\right) \leq \displaystyle \inf_{f \in \mathcal{F}} n_f \leq \begin{cases}
        \textstyle \left(\frac{8m}{\ell} \log \left(\frac{k}{k-\ell}\right) + 1\right) (N\left(\Theta, \varepsilon/(2\ell)\right)-1) & \ell < k \\
        \textstyle \left(\frac{8m}{k} (1 + \log k) + 1\right) (N\left(\Theta, \varepsilon/(2\ell)\right)-1) & \ell = k
    \end{cases}
\]
where the covering and packing numbers are all with respect to $\norm{\cdot}$.
\end{theorem}

We highlight that the upper and lower bounds are off by at most an additional multiplicative $O(\log k)$ factor.
We roughly outline the proof of Theorem \ref{theorem:multi-winner} here. 
To obtain the upper bound, place a minimal $O(\varepsilon/\ell)$-cover on $\Theta$. 
There exists a spanning tree of the elements in the cover in which each edge is of length at most $O(\varepsilon/\ell)$. 
Locate $O(m (H_k - H_{k-\ell}))/ \ell)$ experts (where $H_n$ is the $n$-th harmonic number) uniformly along each edge of the spanning tree. 
This committee has size $N(\Theta, \lVert\cdot\rVert, O(\varepsilon/\ell)) O(m (H_k-H_{k-\ell})/ \ell)$. 
Now, choose $\ell$ candidates as follows. 
First, choose the graph center as we do when we want to choose a single candidate. 
Then, remove this candidate (and her edges) from the graph and choose the graph center of the resulting graph. 
Continue this process until we have chosen $\ell$ candidates. 
Throughout this process, we incur a cumulative regret of $\ell O(\varepsilon/\ell) + O(\varepsilon / (m (H_k-H_{k-\ell}))) (\sum_{i=1}^\ell (m-i)/(k-i+1)) = O(\varepsilon)$. 
The first term bounds the cumulative regret of choosing each candidate over a candidate who was not some expert's favorite. 
The second term bounds the cumulative regret of choosing each candidate over another candidate who was some expert's favorite. 
The sum captures the fact that removing a candidate decreases the connectivity in the remaining graph between each expert's favorite remaining candidate by at most 1 and hence, increases the diameter.
To obtain a nearly matching lower bound, consider our lower bound construction but with a maximal $\Omega(\varepsilon/\ell)$-packing of $\Theta$ and qualities that range from 0 to $2\varepsilon/\ell$ instead. It will follow from the same arguments as those given in Section 5 that if there exists a ball in the packing with fewer than $\Omega(m/k)$ experts, then there exists a set of $m$ candidates such that at least half of the $\ell$ candidates chosen by a deterministic voting rule will have regret $\Omega(\varepsilon/\ell)$ each, so the cumulative regret is at least $\Omega(\varepsilon)$. Thus, any committee that obtains at most $\varepsilon$ cumulative regret requires size $M(\Theta, \lVert\cdot\rVert, \Omega(\varepsilon/\ell))\Omega(m/k)$. 

\subsection{Multi-Round Voting}

In the multi-round setting, we wish to select a single candidate with regret at most $\varepsilon$, but we have access to multiple rounds of voting, in each of which each expert reports only her favorite candidate.
More specifically, before the candidates arrive, we can choose and commit to a (possibly different) committee in each round.
Then, after the candidates arrive, in each round, we elicit the votes from the experts in that round and remove some candidates from consideration.
This setting is quite natural: hiring processes in reality, e.g., faculty recruiting, are often split into multiple rounds.

We show that with two rounds of voting, the number of experts needed to guarantee a regret of at most $\varepsilon$ becomes independent of the number of candidates.
The idea is to introduce a screening process before the selection process.
Intuitively, by introducing a screening process, we can limit the number of candidates to choose from during the selection process.
If we choose a screening committee that is sufficient large, then we can guarantee that the candidates who do not pass the screening process are not much better than the candidates who do pass.
Moreover, if we choose a screening committee whose size only depends on $\varepsilon$, then the number of candidates we have to choose from during the selection process will only depend on $\varepsilon$ as well and will not depend on the number $m$ of candidates who arrived during the screening process.

\begin{theorem}\label{theorem:multi-round-voting}
Let $(V, \norm{\cdot})$ be a normed vector space and $\Theta \subseteq V$.
For all $\varepsilon > 0$, there exists a screening committee of size $N(\Theta, \norm{\cdot}, \varepsilon/2)$ and a selection committee of size at most $8N(\Theta, \norm{\cdot}, \varepsilon/2)^2$ that together guarantee a regret of at most $\varepsilon$ regardless of the number of candidates.
\end{theorem}

\begin{proof}
In the screening round, place an $(\varepsilon/2)$-cover over $\Theta$ of minimum size.
Each candidate who receives a vote in the screening round moves on to the selection round.
Note that the size of the screening committee is $N(\Theta, \norm{\cdot}, \varepsilon/2)$, so there are at most this many candidates who pass the screening round.
Moreover, any candidate who did not receive a vote in this round is at most $\varepsilon/2$ better than a candidate who did receive a vote.

Now, since there are only $N(\Theta, \norm{\cdot}, \varepsilon/2)$ candidates in the selection round, by Theorem \ref{theorem:general-upper-bound}, there exists a committee of size at most $8N(\Theta, \norm{\cdot}, \varepsilon/2)^2$ (take $m \coloneqq N(\Theta, \norm{\cdot}, \varepsilon/2)$ and $\varepsilon \coloneqq \varepsilon/2$) that guarantees a regret of at most $\varepsilon/2$ if we choose the minimal regret candidate.
Thus, after two rounds of voting, we have incurred a cumulative regret of at most $\varepsilon$.
\end{proof}

\bibliographystyle{ACM-Reference-Format}
\bibliography{acm-ec-23}


\appendix

\section{Bounds in Euclidean Space}

To compare our bounds with those derived by \citet{Feldman_Mansour_Nisan_Oren_Tennenholtz_2020}, we derive explicit bounds for the $d$-dimensional unit hypercube $[0,1]^d$ equipped with the $\ell_p$ norm $\norm{\cdot}_p$, the setting studied by \citet{Feldman_Mansour_Nisan_Oren_Tennenholtz_2020}.
We first derive bounds for general convex subsets of $d$-dimensional Euclidean space equipped with any norm using standard volume arguments to bound the covering and packing numbers.

\begin{lemma}[see \citet{yihong}]\label{lemma:yihong's-lemma} 
Let $\norm{\cdot}$ be a norm in $\RR^d$ and $\Theta \subseteq V$ be convex and contain a ball of radius $\varepsilon$.
\[
    \left(\frac{1}{\varepsilon}\right)^d \frac{\mathrm{vol}(\Theta)}{\mathrm{vol}(B(0,1))} \leq N(\Theta, \norm{\cdot}, \varepsilon) \leq M(\Theta, \norm{\cdot}, \varepsilon) \leq \left(\frac{3}{\varepsilon}\right)^d \frac{\mathrm{vol}(\Theta)}{\mathrm{vol}(B(0,1))}
\]
\end{lemma}

\begin{theorem}\label{theorem:combination}
Let $\norm{\cdot}$ be a norm in $\RR^d$, $\varepsilon > 0$, $\Theta \subseteq \RR^d$ be convex and contain a ball of radius $24 \varepsilon$, and $E$ be a universal committee of minimum size that guarantees a regret of at most $\varepsilon$ for any set of $m \geq k(12^d + 1)$ candidates in $\Theta$.
\[
    \Omega\left(\frac{m}{k}\left(\frac{1}{288\varepsilon}\right)^d \frac{\mathrm{vol}(\Theta)}{\mathrm{vol}(B(0,1))} \right) \leq \abs{E} \leq O\left(\frac{m}{k}\left(\frac{6}{\varepsilon}\right)^d \frac{\mathrm{vol}(\Theta)}{\mathrm{vol}(B(0,1))} \right)
\]
\end{theorem}

We highlight that the upper and lower bounds are asymptotically off by no more than a constant raised to the dimension of the underlying space.

\begin{proof}[Proof of Theorem \ref{theorem:combination}]
By Theorem \ref{theorem:general-upper-bound} and Lemma \ref{lemma:yihong's-lemma},
\[
    \abs{E} \leq (8(m-1)/k + 1) (N\left(\Theta, \norm{\cdot}, \varepsilon/2\right) - 1) \leq (8(m-1)/k + 1)\left(\frac{6}{\varepsilon}\right)^d \frac{\mathrm{vol}(\Theta)}{\mathrm{vol}(B(0,1))}
\]
By Theorem \ref{theorem:general-lower-bound} and Lemma \ref{lemma:yihong's-lemma},
\[
    \abs{E} \geq M(\Theta, \norm{\cdot}, 24\varepsilon)\left(\frac{m}{k(N(B(0, 1), \norm{\cdot}, 1/4) + 1)} - 1\right) \geq \left(\frac{1}{24\varepsilon}\right)^d \frac{\mathrm{vol}(\Theta)}{\mathrm{vol}(B(0,1))} \left(\frac{m}{k(12^d + 1)} - 1\right)
\]
\end{proof}

As a corollary of Theorem \ref{theorem:combination}, we obtain the following explicit upper and lower bounds for $([0, 1]^d, \norm{\cdot}_p)$ by bounding the volume of the unit ball with respect to the $\ell_p$ norm.

\begin{theorem}\label{theorem:lp_bounds}
Let $E$ be a universal committee of minimum size that guarantees a regret of at most $\varepsilon \in (0, 1/48)$ for any set of $m \geq 12^d + 1$ candidates in $([0,1]^d, \norm{\cdot}_p)$.
\[
    \Omega\left(m\left(\frac{1}{1152 e^{1/12} \varepsilon \sqrt{\pi}}\right)^d d^{d/p} \sqrt{\frac{d}{p} + 1} \right) \leq \abs{E} \leq O\left(m\left(\frac{3e}{\varepsilon\sqrt{2\pi}}\right)^d d^{d/p} \sqrt{\frac{d}{p} + 1} \right)
\]
\end{theorem}

In contrast, the arguments in \citet{Feldman_Mansour_Nisan_Oren_Tennenholtz_2020} obtain a lower bound of $\Omega(\max\{m, c_1^d \varepsilon^{-d} d^{- d/p}\})$ and an upper bound of $O(c_2^d m^d \varepsilon^{-d} d^{d/p}\sqrt{d/p + 1})$ for some constants $c_1$ and $c_2$.
We now proceed to bound the volume of the unit ball with respect to $\norm{\cdot}_p$, from which Theorem \ref{theorem:lp_bounds} will readily follow.

\begin{lemma}[{\citet[Eq.~5.6.1]{NIST:DLMF}}]\label{lemma:bounds-gamma-x}
For all $x > 0$,
\[
    \sqrt{\frac{2\pi}{x}} \left(\frac{x}{e}\right)^x < \Gamma(x) < \sqrt{\frac{2\pi}{x}} \left(\frac{x}{e}\right)^x e^{\frac{1}{12x}}
\]
\end{lemma}

\begin{lemma}\label{lemma:bounds-gamma-x+1}
For all $x > 0$,
\[
    \frac{\sqrt{2\pi(x+1)}}{e} \left(\frac{x}{e}\right)^x < \Gamma(x+1) < \sqrt{2\pi(x+1)}\left(\frac{x}{e}\right)^x e^{1/12}
\]
\end{lemma}

\begin{proof}
The result follows readily from Lemma \ref{lemma:bounds-gamma-x}.
To see the lower bound, observe that
\[
    \Gamma(x+1) > \sqrt{\frac{2\pi}{x+1}}\left(\frac{x + 1}{e}\right)^{x+1} = \frac{\sqrt{2\pi(x+1)}}{e} \left(\frac{x + 1}{e}\right)^{x} > \frac{\sqrt{2\pi(x+1)}}{e} \left(\frac{x}{e}\right)^{x} 
\]
where the first inequality follows from Lemma \ref{lemma:bounds-gamma-x}.
To see the upper bound, observe that
\begin{align*}
    \Gamma(x+1)
        &< \sqrt{\frac{2\pi}{x + 1}} \left(\frac{x + 1}{e}\right)^{x+1} e^{\frac{1}{12(x+1)}} \tag{Lemma \ref{lemma:bounds-gamma-x}} \\
        &= \frac{\sqrt{2\pi(x+1)}}{e} \left(\frac{x+1}{e}\right)^x e^{\frac{1}{12(x+1)}} \\
        &= \frac{\sqrt{2\pi(x+1)}}{e} \left(\frac{x}{e}\right)^x \left(1 + \frac{1}{x}\right)^x e^{\frac{1}{12(x+1)}} \\
        &\leq \sqrt{2\pi(x+1)} \left(\frac{x}{e}\right)^x e^{1/12}
\end{align*}
\end{proof}

\begin{lemma}\label{lemma:bounds-gamma-1/p-d/p}
For $p \geq 1$,
\begin{align*}
    \frac{\sqrt{2\pi}}{e} \left(\frac{1}{pe}\right)^{1/p} < {} &\Gamma\left(\frac{1}{p} + 1 \right) < 2\sqrt{\pi} \left(\frac{1}{pe}\right)^{1/p} e^{1/12} \\
    \frac{\sqrt{2\pi\left(\frac{d}{p} + 1\right)}}{e} \left(\frac{d}{pe}\right)^{d/p} < {} &\Gamma\left(\frac{d}{p} + 1 \right) < \sqrt{2\pi\left(\frac{d}{p} + 1\right)} \left(\frac{d}{pe}\right)^{d/p} e^{1/12}
\end{align*}
\end{lemma}

\begin{proof}
The upper and lower bounds for $\Gamma(d/p + 1)$ follow immediately from Lemma \ref{lemma:bounds-gamma-x+1}.
The upper and lower bounds for $\Gamma(1/p + 1)$ require one more step each.
\begin{gather*}
    \Gamma\left(\frac{1}{p} + 1 \right) > \frac{\sqrt{2\pi\left(\frac{1}{p} + 1\right)}}{e} \left(\frac{1}{pe}\right)^{1/p} \geq \frac{\sqrt{2\pi}}{e} \left(\frac{1}{pe}\right)^{1/p} \\
    \Gamma\left(\frac{1}{p} + 1 \right) < \sqrt{2\pi\left(\frac{1}{p} + 1\right)} \left(\frac{1}{pe}\right)^{1/p} e^{1/12} \leq 2 \sqrt{\pi} \left(\frac{1}{pe}\right)^{1/p} e^{1/12} \tag{$p \geq 1$}
\end{gather*}
\end{proof}

\begin{lemma}\label{lemma:bounds-volume-unit-ball}
For $p \geq 1$,
\[
    e^{-1} \sqrt{2\pi \left(\frac{d}{p} + 1\right)}\left(\frac{1}{4 e^{1/12} \sqrt{\pi}}\right)^d d^{d/p} < \mathrm{vol}(B(0,1))^{-1} < e^{1/12} \sqrt{2\pi\left(\frac{d}{p} + 1\right)} \left(\frac{e}{2\sqrt{2\pi}}\right)^d d^{d/p} 
\]
\end{lemma}

\begin{proof}
The lemma follows from Lemma \ref{lemma:bounds-gamma-1/p-d/p} and the fact that 
\[
    \mathrm{vol}(B(0,1)) = \frac{\left(2\Gamma\left(\frac{1}{p} + 1\right)\right)^d}{\Gamma\left(\frac{d}{p} + 1\right)}.
\]
\end{proof}

\begin{proof}[Proof of Theorem \ref{theorem:lp_bounds}]
The result follows readily from Theorem \ref{theorem:combination} and Lemma \ref{lemma:bounds-volume-unit-ball}.
\end{proof}

\section{Omitted Algorithms, Lemmas, and Proofs}

\subsection{Algorithms, Lemmas, And Proofs Omitted From Section \ref{section:preliminaries}}

We derive the dual of \ref{primal_lp}.
The Lagrangian of \ref{primal_lp} is
\begin{align*}
    L(q, \lambda)
        = {} & q_h - q_g + \sum_{i \in [n]} \sum_{j \in [m]} \sum_{j' : j \succ_i^k j'} \lambda_{i,j,j'} (q_j - q_{j'} + d(e_i, c_{j'}) - d(e_i, c_j)) \\
        = {} & \sum_{i \in [n]} \sum_{j \in [m]} \sum_{j': j \succ_i^k j'} \lambda_{i,j,j'} (d(e_i, c_{j'}) - d(e_i, c_j)) + q_h \left(1 + \sum_{i \in [n]} \sum_{j' : h \succ_i^k j'} \lambda_{i, h, j'} - \sum_{i \in [n]} \sum_{j : j \succ_i^k h} \lambda_{i, j, h} \right) \\
        & + q_g \left(- 1 + \sum_{i \in [n]} \sum_{j' : g \succ_i^k j'} \lambda_{i, g, j'} - \sum_{i \in [n]} \sum_{j : j \succ_i^k g} \lambda_{i, j, g} \right) + \sum_{\ell \in [m]} q_\ell \left(\sum_{i \in [n]} \sum_{j' : \ell \succ_i^k j'} \lambda_{i, \ell, j'} - \sum_{i \in [n]} \sum_{j : j \succ_i^k \ell} \lambda_{i, j, \ell} \right)
\end{align*}
Note that the dual objective $\max_q L(q, \lambda) = +\infty$ if any of the parenthesized expressions are non-zero.
Thus, the dual problem is
\begin{equation*}\label{dual_lp}
\begin{array}{ll@{}ll}
\min_\lambda 
    & & \displaystyle \sum_{i \in [n]} \sum_{j \in [m]} \sum_{j': j \succ_i^k j'} \lambda_{i,j,j'} (d(e_i, c_{j'}) - d(e_i, c_j)) \\
\text{s.t.} 
    & & \displaystyle \sum_{i \in [n]} \sum_{j : j \succ_i^k h} \lambda_{i, j, h} - \sum_{i \in [n]} \sum_{j' : h \succ_i^k j'} \lambda_{i, h, j'} = 1 \\
    & & \displaystyle \sum_{i \in [n]} \sum_{j' : g \succ_i^k j'} \lambda_{i, g, j'} - \sum_{i \in [n]} \sum_{j : j \succ_i^k g} \lambda_{i, j, g} = 1 \\
    & & \displaystyle \sum_{i \in [n]} \sum_{j' : \ell \succ_i^k j'} \lambda_{i, \ell, j'} - \sum_{i \in [n]} \sum_{j : j \succ_i^k \ell} \lambda_{i, j, \ell} = 0 & \forall\: \ell \not= h,k \\
    & & \lambda_{i, j, j'} \geq 0 & \forall\: i \in [n], j, j' \in [m] \text{ s.t. } j \succ_i^k j'
\end{array}\tag{LP 2}
\end{equation*}

\begin{lemma}\label{lemma:no-negative-cycles}
Let $C:= \{c_1,\ldots, c_m\}$ denote the set of locations for the $m$ candidates, $E:= \{e_1,\ldots,e_n\}$ denote the set of locations for the $n$ experts, and $\succ^k \coloneqq \{\succ_1^k, \dots, \succ_n^k\}$ denote an inducible set of rankings.
$G(C, E, \succ^k)$ has no negative cycles.\footnote{The ideas used to prove this lemma are used by \citet{Feldman_Mansour_Nisan_Oren_Tennenholtz_2020} as well.}
\end{lemma}

\begin{proof}
Let $j_1 \to \dots \to j_T \to j_{T+1} = j_1$ be a cycle in $G(C, E, \succ^k)$.
For all $t \in [T]$, let 
\[
    i_t \coloneqq \arg\min_{i : j_t \succ_i^k j_{t+1}} d(e_i, c_{j_{t+1}}) - d(e_i, c_{j_t}).
\]
Note that $j_\ell \succ_{i_t}^k j_{t+1}$ implies $q_{j_{t+1}} - q_{j_t} \leq  d(e_{i_t}, c_{j_{t+1}}) - d(e_{i_t}, c_{j_t}) = w(j_t, j_{t+1})$, so
\[
    \sum_{t=1}^T w(j_t, j_{t+1}) \geq \sum_{t=1}^{T} q_{j_{t+1}} - q_{j_t} = 0
\]
Thus, there are no negative cycles in $G(C, E, \succ^k)$.
\end{proof}

\begin{lemma}\label{lemma:primal-feasible-solution}
Let $C:= \{c_1,\ldots, c_m\}$ denote the set of locations for the $m$ candidates, $E:= \{e_1,\ldots,e_n\}$ denote the set of locations for the $n$ experts, and $\succ^k \coloneqq \{\succ_1^k, \dots, \succ_n^k\}$ denote an inducible set of rankings.
If expert $i$ ranks $g$ among her top $k$ candidates, then there exists $\alpha < +\infty$ such that setting $\hat{q}_j = d_G(g, j)$ for all $j$ reachable from $g$ and $\hat{q}_j = \alpha - d_G(j, g)$ for all $j$ unreachable from $g$ is consistent with $\succ^k$.
\end{lemma}

\begin{proof}
Let $C_g$ denote the candidates reachable from $g$ in $G(C, E, \succ^k)$.
Since expert $i$ ranks $g$ among her top $k$ candidates, for all $j \not= g$, either $j \succ_i^k g$ or  $j \prec_i^k g$.
But for $j \not\in C_g$, it must be that $j \succ_i^k g$ since otherwise, there would exist an edge from $g$ to $j$, so $j \in C_g$, a contradiction.
It follows that there is an edge from $j$ to $g$, so $d_G(j, g) < +\infty$ for all $j \not\in C_g$ and we can choose $\alpha < +\infty$ to be sufficiently large so that $\alpha \geq  d_G(j, g) + d_G(g,j') - d_G(j,j')$ for all $i \in [n], j \not\in C_g, j' \in C_g$ such that $j \succ_i^k j'$.
Now, consider $i \in [n], j, j' \in [m]$ such that $j \succ_i^k j'$.
Note that either $j, j' \in C_g$, $j, j' \not\in C_g$, or $j \not\in C_g, j' \in C_g$.
Otherwise, there would exist an edge from $j$ to $j'$ in $G(C, E, \succ^k)$, so $j' \in C_g$, a contradiction.
Thus,
\begin{enumerate}[leftmargin=*]
    \item If $j, j' \in C_g$, then $\hat{q}_{j'} = d_G(g,j') \leq d_G(g, j) + d_G(j, j') \leq d_G(g, j) + w(j, j') \leq d_G(g,j) + d(e_i,c_{j'}) - d(e_i,c_{j}) = \hat{q}_j + d(e_i,c_{j'}) - d(e_i,c_{j})$.
    \item If $j, j' \in C \setminus C_g$, then $\hat{q}_{j'} = \alpha - d_G(j', g) \leq \alpha - d_G(j, g) + d_G(j, j') \leq \alpha - d_G(j, g) + w(j, j') \leq \alpha - d_G(j,g) + d(e_i,c_{j'}) - d(e_i,c_{j}) = \hat{q}_j + d(e_i,c_{j'}) - d(e_i,c_{j})$. 
    \item If $j \in C \setminus C_g, j' \in C_g$, $\hat{q}_{j'} = d_G(g, j') \leq \alpha - d_G(j, g) + d_G(j, j') \leq \alpha - d_G(j, g) + w(j, j') \leq \alpha - d_G(j,g) + d(e_i,c_{j'}) - d(e_i,c_{j}) = \hat{q}_j + d(e_i,c_{j'}) - d(e_i,c_{j})$. 
\end{enumerate}
It follows that $\hat{q}$ is consistent with $\succ^k$.
\end{proof}

\begin{proof}[Proof of Lemma \ref{lemma:reduction}]
We now show that if there is no $g \to h$ path in $G(C, E, \succ^k)$, then the optimal objective value of the corresponding \ref{primal_lp} is $+\infty$.
Let $C_g$ denote the candidates reachable from $g$ in $G(C, E, \succ^k)$.
Since $\succ^k$ is inducible, there exists a feasible primal solution $\hat{q}$.
Let $\mathbbm{1}_{C_g} \in \RR^m$ denote the vector with 1's in the coordinates corresponding to the candidates in $C_g$ and 0's everywhere else.
We claim that $\hat{q} - \alpha \mathbbm{1}_{C_g}$ is primal feasible for all $\alpha \geq 0$.
To see why, note that for all $i \in [n], j, j' \in [m]$ such that $j \succ_i^k j'$, either $j, j' \in C_g$, $j, j' \not\in C_g$, or $j \not\in C_g, j' \in C_g$.
Otherwise, there would exist an edge from $j$ to $j'$ in $G(C, E, \succ^k)$, so $j' \in C_g$, a contradiction.
Now, consider $i \in [n], j, j' \in [m]$ such that $j \succ_i^k j'$.
Since $\hat{q}$ is consistent with $\succ^k$, we know that $\hat{q}_j - d(e_i, c_j) \geq \hat{q}_{j'} - d(e_i, c_{j'})$.
Some casework will show that $\hat{q} - \alpha \mathbbm{1}_{C_g}$ is consistent with $\succ^k$ as well for all $\alpha \geq 0$.
\begin{enumerate}[leftmargin=*]
    \item If $j, j' \in C_g$, then $(\hat{q} - \alpha \mathbbm{1}_{C_g})_j - d(e_i, c_j) = (\hat{q}_j - \alpha) - d(e_i, c_j) \geq (\hat{q}_{j'} - \alpha) - d(e_i, c_{j'}) = (\hat{q} - \alpha \mathbbm{1}_{C_g})_{j'} - d(e_i, c_{j'})$.
    \item If $j, j' \not\in C_g$, then $(\hat{q} - \alpha \mathbbm{1}_{C_g})_j - d(e_i, c_j) = \hat{q}_j - d(e_i, c_j) \geq \hat{q}_{j'} - d(e_i, c_{j'}) = (\hat{q} - \alpha \mathbbm{1}_{C_g})_{j'} - d(e_i, c_{j'})$.
    \item If $j \not\in C_g, j' \in C_g$, then $(\hat{q} - \alpha \mathbbm{1}_{C_g})_j - d(e_i, c_j) = \hat{q}_j - d(e_i, c_j) \geq (\hat{q}_{j'} - \alpha) - d(e_i, c_{j'}) = (\hat{q} - \alpha \mathbbm{1}_{C_g})_{j'} - d(e_i, c_{j'})$.
\end{enumerate}
Since there is no $g \to h$ path, $h \not\in C_g$.
Thus, the optimal objective value of the corresponding \ref{primal_lp} is at least $\sup_{\alpha \geq 0} \hat{q}_h - \hat{q}_g + \alpha = +\infty$.

Now, suppose there is a $g \to h$ path in $G(C, E, \succ^k)$.
Let $g = j_1 \to \dots \to j_T = h$ be a shortest $g \to h$ path in $G$.
For all $t \in [T-1]$, let $i_t \coloneqq \arg\min_{i : j_t \succ_i^k j_{t+1}} d(e_i, c_{j_{t+1}}) - d(e_i, c_{j_t})$
and set $\lambda_{i_t, j_t, j_{t+1}} = 1$.
Set the remaining dual variables to 0.
This assignment is a feasible dual solution with objective value
\[
    \sum_{t \in [T-1]} d(e_{i_t}, c_{j_{t+1}}) - d(e_{i_t}, c_{j_t}) = \sum_{t \in [T-1]} w(j_t, j_{t+1}) = d_G(g,h)
\]
where the first equality follows from our choice of $i_t$ and the second from the fact that $j_1 \to \dots \to j_T$ is a shortest $g \to h$ path in $G(C, E, \succ^k)$.
Thus, by weak duality, the optimal objective value of the corresponding \ref{primal_lp} is at most $d_G(g,h)$.

To obtain the reverse inequality, we construct a feasible solution to the corresponding \ref{primal_lp}.
Since there is a $g \to h$ path, there is an edge leaving $g$.
Thus, there exists an expert $i \in [n]$ who ranks $g$ among her top $k$ candidates.
By Lemma \ref{lemma:primal-feasible-solution}, there exists $\alpha < +\infty$ such that setting $\hat{q}_j = d_G(g, j)$ for all $j \in C_g$ and $\hat{q}_j = \alpha - d_G(j, g)$ for all $j \not\in C_g$ is a feasible solution to the corresponding \ref{primal_lp} with objective value $q_h - q_g = d_G(g,h) - d_G(g,g) = d_G(g,h)$.
\end{proof}


\subsection{Algorithms And Proofs Omitted From Section \ref{section:general-upper-bound}}

\begin{proof}[Proof of Corollary \ref{corollary:new-voting-rule}]
Let $v \coloneqq {\succ_1^k}(1)$.
Note that for all $j \not= v$, $v \succ_1^k j$, so there exists an edge from $v$ to $j$ in $G(C, E^*, \succ^k)$. In particular, all candidates are reachable from $v$.
Thus, by Lemma \ref{lemma:primal-feasible-solution}, setting $\hat{q}_j = d_G(v, j)$ for all $j \in [m]$ is consistent with $\succ^k$.
By Lemma \ref{lemma:E^*-low-regret}, $d_G(j^*, j) \leq \varepsilon$ for all $j^* \in \arg\max_{j \in V(\succ^k)} d_G(v, j)$.
By Corollary \ref{corollary:reduction}, the regret of selecting such a candidate is at most $\varepsilon$.

To demonstrate the running time, note that it takes $O(nmk)$ time to construct $G(C, E^*, \succ^k)$.
Moreover, since the number of edges is at most $\abs{\{(j, j') \in [m]^2 : \exists\: i \text{ s.t. } j \succ_i^k j'\}}$, it takes $O(m \cdot \abs{E(G(C, E^*, \succ^k))})$ time to compute $d_G(v, j)$ for all $j \in [m]$.
\end{proof}

\begin{algorithm}[t]
    \SetAlgoNoLine
    \DontPrintSemicolon
    \KwIn{$C \coloneqq \{c_1, \dots, c_m\}$ is a set of $m$ candidate locations,
        $E^* \coloneqq \{e_1, \dots, e_n\}$ is a set of $n$ expert locations as constructed in Construction \ref{construction:upper-bound}, 
        $\succ^k \coloneqq \{\succ^k_1, \dots, \succ^k_n\}$ is a inducible set of rankings}
    \KwOut{a candidate with regret at most $\varepsilon$} 
    
    $G \leftarrow ([m], \{(j, j') \in [m]^2: j \not= j'\}, w(\cdot))$ where $w(j, j') \coloneqq +\infty$ for all $j \not= j'$\;
    \For{$i = 1, \dots, n$}{
        \For{$r = 1, \dots, k$}{
            \For{$r' = j+1, \dots, k$}{
                \If{$d(e_i, c_{{\succ_i^k}(r')}) - d(e_i, c_{{\succ_i^k}(r)}) < w({\succ_i^k}(r), {\succ_i^k}(r'))$}{
                    $w({\succ_i^k}(r), {\succ_i^k}(r')) \leftarrow d(e_i, c_{{\succ_i^k}(r')}) - d(e_i, c_{{\succ_i^k}(r)})$\;
                }
            }
            \For{$j' \in [m] \setminus \{{\succ_i^k}(h)\}_{h=1}^k$}{
                \If{$d(e_i, c_{j'}) - d(e_i, c_{{\succ_i^k}(r)}) < w({\succ_i^k}(r), j')$}{
                    $w({\succ_i^k}(r), j') \leftarrow d(e_i, c_{j'}) - d(e_i, c_{{\succ_i^k}(r)})$\;
                }
            }
        }
    }
    
    Solve the single-source shortest path problem in $G$ with ${\succ_1^k}(1)$ as the source vertex\;
    \KwRet{$\arg\max_{j \in V(\succ^k)} d_G({\succ_1^k}(1), j)$}
    
    \caption{Alternative Voting Rule}
    \label{algorithm:new-voting-rule}
\end{algorithm}

\subsection{Proofs Omitted From Section \ref{section:general-lower-bound}}

\begin{proof}[Proof of Lemma \ref{lemma:bound-on-W-covers}]
Since $N(B(x, r), \norm{\cdot}, r/4) = N(B(0, 1), \norm{\cdot}, 1/4)$, it suffices to show that 
\[
    N(\partial B_\Theta(x, r)[x, \Theta \setminus B_\Theta(0, 12 \varepsilon)], \norm{\cdot}, r/2) \leq N(B(x, r), \norm{\cdot}, r/4)
\]
Let $X \subseteq B(x, r)$ be an $(r/4)$-cover of $B(x, r)$ of size $N(B(x, r), \norm{\cdot}, r/4)$.
Construct an $(r/2)$-cover $X'$ of $\partial B_\Theta(x, r)[x, \Theta \setminus B_\Theta(0, 12 \varepsilon)]$ as follows: for each $x \in X$ such that $\partial B_\Theta(x, r)[x, \Theta \setminus B_\Theta(0, 12 \varepsilon)] \cap B(x, 1/4) \not= \varnothing$, add any $x' \in \partial B_\Theta(x, r)[x, \Theta \setminus B_\Theta(0, 12 \varepsilon)] \cap B(x, 1/4)$ to $X'$.
Note that $\abs{X'} \leq \abs{X} = N(B(x, r), \norm{\cdot}, r/4)$.
Moreover, for any $y \in \partial B_\Theta(x, r)[x, \Theta \setminus B_\Theta(0, 12 \varepsilon)]$, there exists $x \in X, x' \in X'$ such that $\norm{y - x} \leq r/4$ and $\norm{x - x'} \leq r/4$.
Thus, $\norm{y - x'} \leq \norm{y - x} + \norm{x - x'} \leq r/2$ and $X'$ is a $(r/2)$-cover of $\partial B_\Theta(x, r)[x, \Theta \setminus B_\Theta(0, 12 \varepsilon)]$.
Since $N(\partial B_\Theta(x, r)[x, \Theta \setminus B_\Theta(0, 12 \varepsilon)], \norm{\cdot}, r/2)$ is the size of the smallest $(r/2)$-cover of $\partial B_\Theta(x, r)[x, \Theta \setminus B_\Theta(0, 12 \varepsilon)]$, we have that $N(\partial B_\Theta(x, r)[x, \Theta \setminus B_\Theta(0, 12 \varepsilon)], \norm{\cdot}, r/2) \leq \abs{X'} \leq N(B(x, r), \norm{\cdot}, r/4)$
\end{proof}

\begin{proof}[Proof of Lemma \ref{lemma:candidates-lie-in-Theta}]
Since $w \in W \subseteq \partial B_\Theta(x, 4\varepsilon - \norm{x})[x, \Theta \setminus B_\Theta(0, 12 \varepsilon)]$, there exists $y \in \Theta \setminus B_\Theta(0, 12 \varepsilon)$ such that $w = x + \alpha(y - x)$ for some $\alpha \in [0, 1]$.
To show that $x + \lambda (w - x) = x + \alpha \lambda (y-x) \in \Theta$, it suffices to show that $\alpha \lambda \in [0, 1]$ since $\Theta$ is convex and $x, y \in \Theta$.
\[
    0 \overset{(a)}{<} \frac{4\varepsilon - \norm{x}}{\norm{y - x}} = \left(\frac{\norm{w - x}}{\norm{y - x}}\right)\left(\frac{4\varepsilon - \norm{x}}{\norm{w - x}}\right) \overset{(b)}{\leq} \alpha \lambda \overset{(c)}{\leq} \left(\frac{\norm{w - x}}{\norm{y - x}}\right)\left(\frac{4\varepsilon + \norm{x}}{\norm{w - x}}\right) = \frac{4\varepsilon + \norm{x}}{\norm{y - x}} \overset{(d)}{<} 1
\]
(a) and (d) follow from the fact that $\norm{x} < 4\varepsilon$ and the fact that $\norm{y - x} \geq \norm{y} - \norm{x} > 8 \varepsilon$. 
(b) and (c) follow from the fact that $w = x + \alpha(y - x)$ implies $\alpha = \frac{\norm{w - x}}{\norm{y - x}}$ and the fact that $\lambda$ solves $\norm{x + \lambda (w-x)} = 4\varepsilon$ subject to $\lambda \geq 0$, so $\frac{4\varepsilon - \norm{x}}{\norm{w - x}} \leq \lambda \leq \frac{4\varepsilon + \norm{x}}{\norm{w - x}}$.
\end{proof}

\subsection{Proofs Omitted From Section \ref{section:extensions}}

\begin{algorithm}[t]
    \SetAlgoNoLine
    \DontPrintSemicolon
    \KwIn{$C \coloneqq \{c_1, \dots, c_m\}$ is a set of $m$ candidate locations,
        $E \coloneqq \{e_1, \dots, e_n\}$ is a set of $n$ expert locations,
        $\succ^k \coloneqq \{\succ^k_1, \dots, \succ^k_n\}$ is a inducible set of rankings}
    \KwOut{a set of $\ell$ candidates} 
    
    $G \leftarrow ([m], \{(j, j') \in [m]^2: j \not= j'\}, w(\cdot))$ where $w(j, j') \coloneqq +\infty$ for all $j \not= j'$\;
    \For{$i = 1, \dots, n$}{
        \For{$r = 1, \dots, k$}{
            \For{$r' = r+1, \dots, k$}{
                \If{$d(e_i, c_{{\succ_i^k}(r')}) - d(e_i, c_{{\succ_i^k}(r)}) < w({\succ_i^k}(r), {\succ_i^k}(r'))$}{
                    $w({\succ_i^k}(r), {\succ_i^k}(r')) \leftarrow d(e_i, c_{{\succ_i^k}(r')}) - d(e_i, c_{{\succ_i^k}(r)})$\;
                }
            }
            \For{$j' \in [m] \setminus \{{\succ_i^k}(h)\}_{h=1}^k$}{
                \If{$d(e_i, c_{j'}) - d(e_i, c_{{\succ_i^k}(r)}) < w({\succ_i^k}(r), j')$}{
                    $w({\succ_i^k}(r), j') \leftarrow d(e_i, c_{j'}) - d(e_i, c_{{\succ_i^k}(r)})$\;
                }
            }
        }
    }

    $S \gets \varnothing$\;
    \For{$i = 1, \dots, \ell$}{
        $j^* \gets \arg\min_j \max_{j' \not= j} d_G(j,j')$ \tcp*{graph center}
        $S \gets S \cup \{j^*\}$\;
        $G \gets G - j^*$ \tcp*{vertex deletion}
    }
    \KwRet{$S$}
    
    \caption{Multi-Winner Voting Rule}
    \label{algorithm:multi-winner}
\end{algorithm}

\begin{theorem}
Let $(V, \norm{\cdot})$ be a normed vector space and $\Theta \subseteq V$ be convex.
Let $f$ denote the voting rule defined by Algorithm \ref{algorithm:multi-winner}.
There exists a universal committee $E$ of size at most
\[
    \textstyle \left(\frac{8m}{\ell} \min\left\{\log \left(\frac{k}{k-\ell}\right), 1 + \log k\right\} + 1\right) (N\left(\Theta, \norm{\cdot}, \varepsilon/(2\ell)\right)-1)
\]
such that for all sets $C \subseteq \Theta$ of $m$ candidate locations and all quality vectors $q \in \RR^m$, selecting the set $f(C, E, \succ^k)$ of $\ell$ candidates guarantees a cumulative regret of at most $\varepsilon > 0$.
\end{theorem}

\begin{proof}
Define the universal committee $E^*$ as follows.
Let $X \subseteq \Theta$ be an (internal) $(\varepsilon/(2\ell))$-cover of $(\Theta, \norm{\cdot})$ of minimum size, and let $H(X)$ denote the undirected graph with vertices in $X$ and an edge between two vertices if and only if the distance between them is at most $\varepsilon/\ell$:
\[
    H(X) \coloneqq (X, \{\{x,y\} \in X^2:0<d(x,y)\leq \varepsilon/\ell\}
\]
By Lemma \ref{lemma:H-connected}, $H(X)$ is a connected graph, so there exists a spanning tree $T$ of $H(X)$.
For each edge $\{x,y\} \in T$, place an expert at $x + \frac{\ell i}{8m (H_k - H_{k - \ell})}(y-x) \in \Theta$ for $i = 0, 1, \dots, H_k - H_{k - \ell}$ where $H_n$ is the $n$-th harmonic number.
Note that the distance between two consecutive experts along the line segment between $x$ and $y$ is at most $\frac{\ell}{8m (H_k - H_{k - \ell})} \cdot \frac{\varepsilon}{\ell} = \frac{\varepsilon}{8m (H_k - H_{k - \ell})}$.
Moreover, 
\begin{align*}
    \abs{E^*} 
        &\leq \textstyle \left(\frac{8m}{\ell} (H_k - H_{k - \ell}) + 1\right)(N(\Theta, \norm{\cdot}, \varepsilon / (2\ell)) - 1) \\
        &\leq \begin{cases}
        \textstyle \left(\frac{8m}{\ell} \log \left(\frac{k}{k-\ell}\right) + 1\right) (N\left(\Theta, \varepsilon/(2\ell)\right)-1) & \ell < k \\
        \textstyle \left(\frac{8m}{k} (1 + \log k) + 1\right) (N\left(\Theta, \varepsilon/(2\ell)\right)-1) & \ell = k
    \end{cases}
\end{align*}

Now, let $C = \{c_1, \dots, c_m\}$ denote the set of locations for the $m$ candidates and $\succ^k = (\succ^k_1, \dots, \succ^k_n)$ denote the reported ranking of the experts in $E^*$.
By the same argument as in the proof of Lemma \ref{lemma:remains-connected}, for all $J \subseteq [m]$ such that $\abs{J} \leq k - 1$ and all $j \not= j' \in V(\succ^k) \setminus J$, there exists a $j \to j'$ path in $G(C,E^*,\succ^k)[\frac{\varepsilon}{8m(H_k - H_{k - \ell})}] - J$.
Thus, by the same argument as in the proof of Lemma \ref{lemma:bound-on-diameter}, for all $J \subseteq [m]$ such that $\abs{J} \leq k - 1$,
\[
    \textstyle \diam{G(C,E^*,\succ^k)\left[\frac{\varepsilon}{8m(H_k - H_{k - \ell})}\right] - J, V(\succ^k) \setminus J} \leq \frac{2(m-\abs{J} - 1)}{k - \abs{J}} 
\]

Let $\hat{q} \in \RR^m$ be a quality vector consistent with $\succ^k$.
Let $J_1 \coloneqq \varnothing$, $j_1 \in \arg\max_{j \in V(\succ^k)} \hat{q}_j$, and $j^*_1$ be the graph center of $G(C,E^*,\succ^k)$.
For $i = 2, \dots, \ell$, let $J_i \coloneqq \{j^*_1, \dots, j^*_{i-1}\}$, $j_i \in \arg\max_{j \in V(\succ^k) \setminus J_i} \hat{q}_j$, and $j^*_i$ be the graph center of $G(C,E^*, \succ^{k}) - J_i$.
\begin{align*}
    \max_{j \in [m] \setminus J_i} d_{G(C,E^*, \succ^{k}) - J_i} (j^*_i, j) 
        & \leq \max_{j \in [m] \setminus J_i} d_{G(C,E^*, \succ^{k}) - J_i} (j_i, j) \tag{definition of graph center} \\
        &\leq \textstyle \frac{\varepsilon}{2\ell} + 2 \frac{\varepsilon}{8m(H_k - H_{k - \ell})} \diam{G(C,E^*,\succ^k)\left[\frac{\varepsilon}{8m(H_k - H_{k - \ell})}\right] - J_i, V(\succ^k) \setminus J_i} \tag{Lemma \ref{lemma:main-lemma}} \\
        &\leq  \textstyle \frac{\varepsilon}{2\ell} + 2 \frac{\varepsilon}{8m(H_k - H_{k - \ell})} \frac{2(m - i)}{k - i + 1} 
\end{align*}

For $i = 1, \dots, \ell$, let $j_{(i)} \in [m]$ denote the candidate with the $i$-th highest quality.
If $\{j_{(1)}, \dots, j_{(\ell)}\} \cap J_\ell = \varnothing$, then the regret of selecting $f(C, E^*, \succ^k)$ is
\begin{align*}
    \sum_{i=1}^\ell (q_{j_{(i)}} - q_{j^*_i})
        &\leq \sum_{i=1}^\ell d_{G(C,E^*, \succ^{k})} (j^*_i, j_{(i)}) \tag{weak duality} \\
        &\leq \sum_{i=1}^\ell d_{G(C,E^*, \succ^{k}) - J_i} (j^*_i, j_{(i)}) \\
        &\leq \sum_{i=1}^\ell \max_{j \in [m] \setminus J_i} d_{G(C,E^*, \succ^{k}) - J_i} (j^*_i, j) \\
        &\leq \sum_{i=1}^\ell \left(\frac{\varepsilon}{2\ell} + 2 \frac{\varepsilon}{8m(H_k - H_{k - \ell})} \frac{2(m - i)}{k - i + 1}\right) \\
        &= \frac{\varepsilon}{2} + \frac{\varepsilon}{4m(H_k - H_{k - \ell})} \sum_{i=1}^\ell\frac{2(m - i)}{k - i + 1} \\
        &\leq \varepsilon
\end{align*}
If $\{j_{(1)}, \dots, j_{(\ell)}\} \cap J_\ell \not= \varnothing$, then carry out the above argument for the disjoint union and note that we incur no regret by choosing the candidates in the intersection.
\end{proof}

\begin{theorem}
Let $(V, \norm{\cdot})$ be a normed vector space, $\varepsilon > 0$, and $\Theta \subseteq V$ be convex with $\mathrm{diam}(\Theta) \geq 12\varepsilon$.
Let $\mathcal{F}$ denote the set of deterministic voting rules that take as inputs candidate locations, expert locations, and top $k$ rankings and output a set of $\ell$ candidates.
Any universal committee that guarantees a cumulative regret of at most $\varepsilon$ for any set of $m \geq k(N(B(0,1), \norm{\cdot}, 1/4) + 1)$ candidates in $\Theta$ under some $f \in F$ is of size at least
\[
    \textstyle M(\Theta, \norm{\cdot}, 24\varepsilon/\ell)\left(\frac{m}{k(N(B(0,1), \norm{\cdot}, 1/4) + 1)} - 1\right)
\]
\end{theorem}

\begin{proof}
The proof is identical as that of Theorem \ref{theorem:general-lower-bound}; simply replace $\varepsilon$ with $\varepsilon/\ell$ and note that in each world, the best set of $\ell$ candidates has cumulative quality $2\varepsilon$ but in one of the worlds, the set of $\ell$ candidates we choose has cumulative quality $\min\{2\varepsilon - \frac{1}{2}\sum_{i=1}^\ell \norm{x_i}, \frac{1}{2} \sum_{i=1}^\ell \norm{x_i}\} \leq \varepsilon$ (where each $\norm{x_i} \leq 4 \varepsilon / \ell$).
\end{proof}

\end{document}